\documentclass[journal,10pt]{IEEEtran}
\usepackage{booktabs}
\usepackage{multirow}
\usepackage{makecell}
\usepackage{siunitx}

\newcommand{\dataset}[3]{%
  \multirow{4}{*}{%
    \makecell[l]{\texttt{#1}\\$(#2,#3)$}%
  }%
}
\usepackage{amsfonts,amssymb,mathtools}
\usepackage{amsthm}
\usepackage{bm}
\usepackage{booktabs,array}
\usepackage{algorithm}
\usepackage{algpseudocode}
\usepackage[shortlabels]{enumitem}
\usepackage[colorlinks=true,allcolors=blue]{hyperref}
\usepackage[nameinlink,noabbrev]{cleveref}
\usepackage{xcolor}
\usepackage{bm}
\usepackage{graphicx}
\usepackage[caption=false,font=footnotesize]{subfig}
\usepackage{cite}
\usepackage{amsmath,amssymb}
\usepackage{mathtools}

\theoremstyle{plain}
\newtheorem{theorem}{Theorem}
\newtheorem{lemma}{Lemma}
\newtheorem{proposition}{Proposition}

\theoremstyle{definition}
\newtheorem{assumption}{Assumption}

\theoremstyle{remark}
\newtheorem{remark}{Remark}

\DeclareMathOperator{\prox}{prox}
\DeclareMathOperator{\sign}{sign}

\DeclareMathOperator{\proj}{proj}

\newcommand{\norm}[1]{\left\lVert #1\right\rVert}
\newcommand{\abs}[1]{\left\lvert #1\right\rvert}
\newcommand{\set}[1]{\left\{#1\right\}}

\newcommand{\paren}[1]{\left(#1\right)}

\title{Joint Structure Identification and Newton Acceleration via Proximal Line Search for Nonconvex Optimization}
\author{Yiming~Zhou and Wei~Dai%
\thanks{Y. Zhou and W. Dai are with the Department of Electrical and
Electronic Engineering, Imperial College London, London, UK 
(e-mail: yiming.zhou18@imperial.ac.uk;
wei.dai1@imperial.ac.uk).}%
}

\begin{document}
\maketitle

\begin{abstract}
We consider composite optimization with a smooth, possibly nonconvex function and a separable convex polyhedral term. Such problems have a simple nonsmooth structure but arise widely in applications. Existing Newton-type methods may fail to identify nonsmooth coordinates during direction computation, move coordinates away from them reached by the full step during line search, or rely on separate first-order steps for identification. In this work, we propose a unified inexact proximal Newton method that couples structure identification with second-order acceleration through a proximal line search designed to retain the nonsmooth structure revealed at the local model. In particular, on a predicted working set, we solve a quadratic model inexactly but enforce its constraints exactly, allowing coordinates to reach new kinks without leaving the selected affine intervals. This set also includes small-margin kink coordinates to exploit Hessian coupling effect. To globalize the direction, we construct a proximal line search using an isotropic proximal model whose linear term is chosen so that the initial trial exactly recovers the full step. Backtracking increases the proximal curvature until sufficient decrease holds, allowing coordinates to remain at kinks reached by the full step while damping other components. Generating each trial point requires only proximal operator evaluations. We prove finite backtracking and finite active-set identification under strict complementarity and a non-singular reduced Hessian, together with local Q-superlinear convergence under a directional Dennis–Moré condition and vanishing relative inexactness. Numerical results verify the effectiveness of the proposed method.


\end{abstract}

\vspace{-1em}
\section{Introduction}
\label{sec:I}
\IEEEPARstart{I}{n} this paper, we consider the composite optimization
\begin{equation}
    \min_{\bm{x} \in \mathbb{R}^n} ~ F(\bm{x}) := q(\bm{x})  +\sum_{i=1}^n h_i(x_i),
    \label{eq:problem-formulation}
\end{equation}
where $q:\mathbb{R}^n \to \mathbb{R}$ is smooth and possibly nonconvex, and each
$h_i:\mathbb{R} \to \mathbb{R}$ is convex polyhedral. Such structured nonsmooth formulations arise widely in signal processing and
machine learning. Examples include $\ell_1$ regularization and its
weighted variants, which form the basis of many sparse recovery methods for identifying low-dimensional structure in high-dimensional data \cite{candes2008enhancing}.  More problems can be cast as \eqref{eq:problem-formulation} through suitable objective decompositions \cite{yao2018efficient} or changes of variables. For example, finite-slope folded-concave sparsity penalties admit exact decompositions into a weighted \(\ell_1\) term and a nonconvex smooth term. Consider $R(\bm x)=\sum_{i=1}^n \phi_i\left(\left|x_i\right|\right)$, where each $\phi_i$ is nondecreasing and concave, has a finite right derivative $\phi_{i,+}^{\prime}(0)$, and has a Lipschitz-continuous derivative. Defining $\psi_i(t) \triangleq \phi_{i,+}^{\prime}(0)|t|-$ $\phi_i(|t|)$ gives the exact decomposition
\begin{equation}
\tilde{q}(\bm x)+\sum_{i=1}^n \phi_i\left(\left|x_i\right|\right)=\underbrace{\left(\tilde{q}(\bm x)-\sum_{i=1}^n \psi_i\left(x_i\right)\right)}_{q(\bm x)}+\underbrace{\sum_{i=1}^n \phi_{i,+}^{\prime}(0)\left|x_i\right|}_{\sum_{i=1}^n h_i\left(x_i\right)} .\nonumber
\end{equation}
Table~\ref{tab:penalty-splits} summarizes representative choices of \(\phi_{i,+}^{\prime}(0)\) and \(\psi_i\). Suitable changes of variables allow SVMs \cite{rybak2025inference} and quantile regression models \cite{wang2024optimal} to be cast in the form \eqref{eq:problem-formulation}. Consider
\begin{equation}
\min_{\bm{\theta}\in\mathbb{R}^{d}}
\frac{1}{2}\bm{\theta}^{\top}\bm P\bm{\theta}
+\sum_{i=1}^{n}
h_i\!\left(c_i+(\bm A\bm{\theta})_i\right),
\nonumber
\end{equation}
where $\bm P\succ 0$, $\bm A\in\mathbb{R}^{n\times d}$ has full row rank, and $\bm c\in\mathbb{R}^{n}$. Setting
$\bm x=\bm c+\bm A\bm{\theta}$ and
$\bm G=\bm A\bm P^{-1}\bm A^{\top}\succ 0$, and eliminating
$\bm{\theta}$, gives the formulation with the residual variable
\begin{equation}
\min_{\bm x\in\mathbb{R}^{n}}
\frac{1}{2}(\bm x-\bm c)^{\top}
\bm G^{-1}(\bm x-\bm c)
+\sum_{i=1}^{n}h_i(x_i).
\nonumber
\end{equation}
The solution $\bm x^\star$ recovers the original variable as $\bm{\theta}^\star
=\bm P^{-1}\bm A^{\top}\bm G^{-1}
 (\bm x^\star-\bm c).$
For a binary hinge-loss SVM, take
$\bm c=\bm 0$, $\bm A=\bm Y\bm Q$, and
$h_i(x)=C_i[1-x]_+$, where
$\bm Y=\operatorname{diag}(y_1,\ldots,y_n)$. For
ridge-regularized quantile regression, take
$\bm c=\bm y$, $\bm A=-\bm Q$, and $h_i(x)=\omega_i(\tau[x]_+ +(1-\tau)[-x]_+)$.
In this case, $\bm x=\bm y-\bm Q\bm{\theta}$ is the residual variable. 
\begin{table}[t]
\caption{Representative exact splits
$\phi_i(|t|)=\phi_{i,+}^{\prime}(0)|t|-\psi_i(t)$.}
\label{tab:penalty-splits}
\vspace{-1em}
\centering
\footnotesize
\setlength{\tabcolsep}{2.5pt}
\renewcommand{\arraystretch}{1.15}
\begin{tabular}{@{}lcl@{}}
\toprule
Penalty $\phi_i$ & $\phi_{i,+}^{\prime}(0)$ &  $\psi_i(t)$ \\
\midrule
SCAD \cite{fan2001variable}, $a_i>2$
& $\mu_i$
& $\displaystyle
\begin{cases}
0,
    & |t|\leq \mu_i,\\
\dfrac{(|t|-\mu_i)^2}{2(a_i-1)},
    & \mu_i<|t|\leq a_i\mu_i,\\
\mu_i |t|-\dfrac{(a_i+1)\mu_i^2}{2},
    & |t|>a_i\mu_i
\end{cases}$ \\

\addlinespace[2pt]
MCP \cite{zhang2010nearly}, $\gamma_i>0$
& $\mu_i$
& $\displaystyle
\begin{cases}
\dfrac{|t|^2}{2\gamma_i},
    & |t|\leq \gamma_i\mu_i,\\
\mu_i |t|-\dfrac{\gamma_i\mu_i^2}{2},
    & |t|>\gamma_i\mu_i
\end{cases}$ \\

\addlinespace[2pt]
CEL0 \cite{soubies2015continuous}, $d_i>0$
& $d_i\sqrt{2\mu}$
& $\displaystyle
\begin{cases}
\dfrac{d_i^2|t|^2}{2},
    & |t|\leq \sqrt{2\mu}/d_i,\\
d_i\sqrt{2\mu}\,|t|-\mu,
    & |t|>\sqrt{2\mu}/d_i
\end{cases}$ \\
\bottomrule
\end{tabular}
\end{table}

\subsection{Related Works}
\label{subsec:I-A}
The proximal gradient method is the first-order baseline to \eqref{eq:problem-formulation}. Each iteration evaluates \(\nabla q\) and applies the Euclidean proximal operator of \(h\). Common acceleration strategies include adaptive step size \cite{bonettini2016variable} and inertial updates \cite{ochs2019adaptive,liang2023average,bonettini2024heavyball}. These methods are inexpensive per iteration and effective at scale for moderate accuracy, but ill-conditioning can limit high-accuracy convergence, motivating the use of second-order information. We next turn to second-order methods and group them according to how they compute Newton-type directions. 

Proximal Newton-type methods compute the direction by approximately solving
\begin{equation}
   \bm{p}_k \in
\arg \min_{\bm{p}}~\left\langle\nabla q\left(\bm x_k\right), \bm p\right\rangle+\frac{1}{2} \bm p^{\top} \bm B_k \bm p+h\left(\bm x_k+\bm{p}\right),\nonumber
\end{equation}
where $\bm B_k$ is either the exact Hessian $\nabla^2 q\left(\bm x_k\right)$ or its suitable approximation, and $h$ is used in the exact form. Existing methods differ mainly in their globalization strategies. A line-search approach computes the direction satisfying relative proximal residual and model decrease conditions, switches to proximal gradient if the direction fails a descent test, and then applies an Armijo line search \cite{kanzow2021globalized}. Another approach avoids line search by regularizing the Hessian, assesses the full inexact step using the ratio of actual to predicted reduction, and updates the regularization accordingly \cite{vomdahl2024inexact}. Trust-region methods instead constrain the step explicitly and update the radius based on model agreement \cite{aravkin2022trustregion,baraldi2023trustregion}. However, generic model-gap and KKT-residual inexact criteria control only subproblem accuracy and do not ensure that the endpoint identifies the active polyhedral face. For $\ell_1$-regularized problem, the resulting endpoint need not be sparse. The method in \cite{lee2023accelerating} achieves active-structure identification by computing an inexact direction with a proximal inner solver under a fixed relative model-gap tolerance. The overall algorithm is two-stage, requiring a switch to smooth optimization on the identified manifold and proximal gradient steps as a fallback. 

Another class of Newton-type methods exploits the structure of \(h(\bm{x}_k)\) when computing the search direction. These methods typically use \(h(\bm{x}_k)\) to select a candidate orthant and solve a Newton model for the associated smooth restriction. Orthant consistency is ensured by post-processing  the direction and projection line search. The resulting direction subproblem has the generic form
\begin{equation}
\bm{p}_k \in
\arg\min_{\bm{p}\in S_k}
\left\{
\left\langle \widetilde{\bm{g}}_k,\bm{p}\right\rangle
+\frac{1}{2}\bm{p}^{\top}\bm{B}_k\bm{p}
\right\}.\nonumber
\end{equation}
Here, \(S_k\) denotes either the full space or a selected free subspace, and \(\widetilde{\bm{g}}_k\) is the pseudo-gradient associated with the selected orthant. Specifically, the work in \cite{andrew2007scalable} projects the trial points generated during the line search onto the selected orthant. The semismooth Newton framework of \cite{byrd2016family} characterizes orthant methods as reduced-space minimization followed by a projected line search. A related method accommodates a nonconvex smooth term and augments the direction system with diagonal curvature information derived from a partial Huber regularization of the \(\ell_1\) term \cite{delosreyes2017secondorder}. These methods are primarily designed for \(\ell_1\)-regularized problems. During line search, applying a uniform step size can leave a coordinate nonzero even when the full projected step would set it to zero, losing the structure revealed at the model endpoint.

Active-set Newton methods generally separate structure identification from second-order acceleration. An identification phase estimates the active manifold, after which Newton-type steps are applied to the resulting smooth reduced problem \cite{lewis2021active}. For convex composite objectives, forward-backward methods identify this manifold in finitely many iterations under partial smoothness and nondegeneracy \cite{liang2017activity}. For nonconvex problems, the method interleaves proximal-gradient steps that update the manifold estimate with Riemannian Newton steps that accelerate optimization on the current manifold \cite{bareilles2023newton}. For specific problems, i.e., \(\ell_1\)-regularized optimization, related methods either combine active-set prediction with inexact second-order subspace minimization \cite{keskar2016secondorder,cheng2021active} or use reduced proximal-gradient steps to expand the predicted support \cite{chen2017reduced}. Thus, although second-order acceleration is effective after stabilization, identification and Newton acceleration remain staged or alternating, and first-order steps or tests are needed to establish and revise the active set.
\subsection{Our contribution}
Although existing second-order methods can substantially improve the efficiency for solving \eqref{eq:problem-formulation}, they do not always make full use of the available structure. This motivates a
\emph{unified algorithm for joint active-set identification and
Newton-type acceleration via proximal line search}. We summarize our
main contributions below.
\begin{itemize}
    \item \emph{Unified identification and acceleration via proximal line
search.}
We develop a single-stage proximal Newton method that solves a quadratic model inexactly but enforces its box constraints exactly, allowing coordinates to reach new kinks within selected affine intervals. The working set also admits small-margin kink coordinates to exploit Hessian coupling. To keep the structure obtained by this model, we design the linear term of an isotropic proximal model so that the initial trial exactly recovers the computed full step. Backtracking increases proximal curvature until sufficient decrease holds, allowing endpoint kink coordinates to remain fixed while other components are damped. Trial points require only separable proximal evaluations.

\item \emph{Global convergence and finite identification.} We show the proposed method automatically reduces to an inexact Newton method on the smooth subspace. Under basic regularity and relative inexactness conditions, we prove finite termination of backtracking and stationarity of every accumulation point. Strict complementarity and nonsingularity of the reduced Hessian further ensure convergence of the full sequence and finite active-set identification. Under a directional Dennis--Moré condition and vanishing relative inexactness, full steps are eventually accepted and the iterates converge locally \(Q\)-superlinearly.

\end{itemize}

\subsection{Notations}
\label{subsec:notation}
Let $[n]:=\{1,\ldots,n\}$, $\mathbb{R}_{+}:=[0,\infty)$, and
$\mathbb{R}_{++}:=(0,\infty)$. Vectors and matrices are written in boldface,
and $(\cdot)^{\top}$ denotes transposition. For index sets
$\mathcal I,\mathcal J\subseteq[n]$, $|\mathcal I|$ and
$\mathcal I^{c}$ denote the cardinality and complement of $\mathcal I$;
$\bm{x}_{\mathcal I}$ and $\bm{A}_{\mathcal I\mathcal J}$ denote the
corresponding subvector and submatrix. The coordinate embedding
$\bm{E}_{\mathcal I}$ satisfies
$\bm{E}_{\mathcal I}^{\top}\bm{x}=\bm{x}_{\mathcal I}$, and
$\bm{e}^{(i)}$ is the $i$th canonical basis vector. The Euclidean inner product and norm are denoted by
$\langle\cdot,\cdot\rangle$ and $\|\cdot\|$, respectively. For matrices,
$\|\cdot\|$ is the spectral norm, while $\|\cdot\|_{p}$ denotes the usual
$\ell_p$ norm. For $\bm{A}\succ0$, set
$\|\bm{u}\|_{\bm{A}}:=(\bm{u}^{\top}\bm{A}\bm{u})^{1/2}$ and let
$\operatorname{dist}_{\bm{A}}(\bm{z},\mathcal C)$ be the associated distance.
We use $\lambda_{\min}(\bm{A})$, $\lambda_{\max}(\bm{A})$, and
$\sigma_{\min}(\bm{A})$ for the extreme eigenvalues and the smallest singular
value, and write $\bm{A}\succeq\bm{B}$ when $\bm{A}-\bm{B}$ is positive
semidefinite. For a nonempty closed convex set $\mathcal C$, $\Pi_{\mathcal C}$,
$\operatorname{dist}(\cdot,\mathcal C)$, $N_{\mathcal C}$,
$\iota_{\mathcal C}$, and $\operatorname{ri}(\mathcal C)$ denote its Euclidean
projection, distance, normal cone, indicator function, and relative interior,
respectively. For a function $\varphi$, $\partial\varphi$ denotes its limiting
subdifferential. $[t]_{+}:=\max\{t,0\}$. Finally,
$h(\bm{x}):=\sum_{i=1}^{n}h_i(x_i)$,
$\bm{g}(\bm{x}):=\nabla q(\bm{x})$, and
$\bm{H}(\bm{x}):=\nabla^{2}q(\bm{x})$ whenever the Hessian exists; the
subscript $k$ indicates evaluation at $\bm{x}_k$.

\section{Preliminaries}
\begin{assumption}
\label{ass:main-1}
    The following conditions hold for \eqref{eq:problem-formulation}.
    \begin{enumerate}[label=(\roman*)]
        \item The function $q: \mathbb{R}^n \rightarrow \mathbb{R}$ is smooth, possibly nonconvex with a $L$-Lipschitz continuous gradient
        \begin{equation}
\|\nabla q(\bm x)-\nabla q(\bm y)\| \leq L\|\bm x-\bm y\|, \quad \forall \bm x, \bm y \in \mathbb{R}^n .
\end{equation}
\item For every $i \in[n]$, the function $h_i: \mathbb{R} \rightarrow \mathbb{R}$ is finite-valued, convex, and polyhedral.
\item The function $F$ is bounded below, and the initial sublevel set
\begin{equation}
\mathcal{L}_0:=\left\{\bm{x} \in \mathbb{R}^n: F(\bm{x}) \leq F\left(\bm{x}_0\right)\right\} 
\end{equation}
is compact.
    \end{enumerate}
\end{assumption}
Assumption \ref{ass:main-1} (i) gives the standard quadratic upper model
\begin{equation}
q(\bm{y}) \leq q(\bm{x})+\langle\bm{g}(\bm{x}), \bm{y}-\bm{x}\rangle+\frac{L}{2}\|\bm{y}-\bm{x}\|^2, ~ \bm{x}, \bm{y} \in \mathbb{R}^n .
\end{equation}
Condition (ii) implies $h$ and $F$ are continuous on $\mathbb{R}^n$, and (iii) ensures $F$  attains its minimum over $\mathcal{L}_0$. For $\sigma>0$, the proximal operator of $h$ is
\begin{equation}
    \operatorname{prox}_{\frac{1}{\sigma}h}(\bm{z})
    \coloneqq
    \underset{\bm{u}\in\mathbb{R}^{n}}
    {\operatorname{argmin}}
    \left\{
        h(\bm{u})
        +
        \frac{\sigma}{2}
        \|\bm{u}-\bm{z}\|^{2}
    \right\}.
    \label{eq:proximal-map}
\end{equation}
The objective in \eqref{eq:proximal-map} is strongly convex so the proximal operator is single-valued and nonexpansive. Separability gives
\begin{equation}
    \left[
        \operatorname{prox}_{\frac{1}{\sigma}h}(\bm{z})
    \right]_{i}
    =
    \operatorname{prox}_{\frac{1}{\sigma}h_i}(z_i),
    ~i\in[n].
    \label{eq:coordinate-proximal-map}
\end{equation}
\subsection{Polyhedral structure of $h_i$}
For each scalar convex function $h_i$, the left and right derivatives exist at every point. For arbitrary $t \in \mathbb{R}$, we denote them by $h'_{i,-}(t)$ and
 $h'_{i,+}(t)$, respectively. Define the kink point set $\mathcal{B}:=\left\{t \in \mathbb{R}: h'_{i,-}(t)<h'_{i,+}(t)\right\} .$ Moreover,
$t\notin\mathcal{B}$ if and only if $h'_{i,-}(t)=h'_{i,+}(t)$. The subdifferential of \(h\) at $t$ is $\partial h_i(t)=[a_i^{-}(t),a_i^{+}(t)].$ Consequently, 
\begin{equation}
    \partial h(\bm{x})=\prod_{i=1}^n\left[h'_{i,-}(x_i), h'_{i,+}(x_i)\right].
    \label{eq:partial-h-1}
\end{equation}
The directional derivative of $h$ is therefore
\begin{equation}
h^{\prime}(\bm{x} ; \bm{d})=\sum_{i: d_i>0} h'_{i,+}\left(x_i\right) d_i+\sum_{i: d_i<0} h'_{i,-}\left(x_i\right) d_i .
\label{eq:h-prime}
\end{equation}
For $\xi \in\{-1,+1\}$, we write $h'_{i,\xi}$ to collect $h'_{i,+}$ and $h'_{i,-}$. Then
\begin{equation}
F^{\prime}\left(\bm{x} ; \xi \bm{e}_i\right)=\xi\left(g_i(\bm{x})+h'_{i,\xi}\left(x_i\right)\right) .
\label{eq:F-prime}
\end{equation}

\subsection{Stationarity, residuals, and active sets}
A point $\bm{x}$ is called stationary if $\bm{0} \in \partial F(\bm{x})$, and the stationary set is $\mathcal{X} := \left\{ \bm{x}\in\mathbb{R}^{n} : \bm{0}\in\partial F(\bm{x}) \right\}$. The following  stationarity conditions are equivalent. 
\begin{proposition}
\label{prop:1}
    For $\bm x \in \mathbb{R}$, the following statements are equivalent.
    \begin{enumerate}[label=(\roman*)]
        \item $\bm{x} \in \mathcal{X}$.
        \item For every $i \in[n]$, $-g_i(\bm{x}) \in\left[h'_{i,-}\left(x_i\right), h'_{i,+}\left(x_i\right)\right]$.
\item For every $i \in[n]$, $F^{\prime}\left(\bm{x} ; \xi \bm{e}_i\right) \geq 0$ for $\xi \in \{-1,+1\}$.
\item $F^{\prime}(\bm{x} ; \bm{d}) \geq 0$ for every $\bm{d} \in \mathbb{R}^n$.
    \end{enumerate}
\end{proposition}
\begin{proof}
    The equivalence of (i) and (ii) follows from \eqref{eq:partial-h-1} and the definition of $\partial F$.  The equivalence of (ii) and (iii) follows from \eqref{eq:F-prime}. Finally, \eqref{eq:h-prime} gives the equivalence with (iv).
\end{proof}
At a kink point, the two-sided condition in (iii) is required to establish stationarity. Moreover, \(F^{\prime}(\bm{x};\xi\bm{e}_i)<0\) indicates that \(\xi\bm{e}_i\) is a strict descent direction. We use two complementary stationarity residuals. The first
is the minimum-norm subgradient
\begin{equation}
    \begin{aligned}
        \bm{r}(\bm{x})
        &\coloneqq
        \Pi_{\partial F(\bm{x})}(\bm{0})=
        \bm{g}(\bm{x})
        +
        \Pi_{\partial h(\bm{x})}
        \bigl(
            -\bm{g}(\bm{x})
        \bigr).
    \end{aligned}
    \label{eq:minnorm-residual}
\end{equation}
It satisfies
\begin{equation}
    \|\bm{r}(\bm{x})\|
    =
    \min_{\bm{v}\in\partial F(\bm{x})}
    \|\bm{v}\|
    =
    \operatorname{dist}
    \bigl(
        \bm{0},
        \partial F(\bm{x})
    \bigr),
    \label{eq:minnorm-residual-norm}
\end{equation}
and its components are
\begin{equation}
    r_i(\bm{x})
    =
    g_i(\bm{x})
    +
    \Pi_{
        [h_{i,-}(x_i),\,h_{i,+}(x_i)]
    }
    \bigl(
        -g_i(\bm{x})
    \bigr),
    ~ i\in[n].
    \label{eq:minnorm-residual-coordinate}
\end{equation}
$\bm{r}(\bm{x})$ vanishes exactly on
$\mathcal{X}$. In addition
\begin{equation}
    \begin{aligned}
        |r_i(\bm{x})|
        &=
        \operatorname{dist}
        \left(
            -g_i(\bm{x}),
            [h'_{i,-}(x_i),h'_{i,+}(x_i)]
        \right)
        \\
        &=
        \left[
            -F'\bigl(
                \bm{x};
                \bm{e}_i
            \bigr)
        \right]_{+}
        +
        \left[
            -F'\bigl(
                \bm{x};
                -\bm{e}_i
            \bigr)
        \right]_{+}.
    \end{aligned}
    \label{eq:residual-directional-violation}
\end{equation}
The two positive-part terms in \eqref{eq:residual-directional-violation} cannot both be nonzero because $F'\bigl(
        \bm{x};
        \bm{e}_i
    \bigr)
    +
    F'\bigl(
        \bm{x};
        -\bm{e}_i
    \bigr)
    =
    h'_{i,+}(x_i)-h'_{i,-}(x_i)
    \geq 0$. $|r_i(\bm{x})|$ measures the magnitude of the violated one-sided stationarity condition. $\bm{r}(\bm{x})$ is discontinuous when a coordinate crosses a kink point. For a continuous stationarity measure, define the proximal-gradient residual
    \begin{equation}
    \bm{G}_{\sigma}(\bm{x})
    \coloneqq
    \sigma
    \left[
        \bm{x}
        -
        \operatorname{prox}{\sigma^{-1}h}
        \left(
            \bm{x}
            -
            \sigma^{-1}\bm{g}(\bm{x})
        \right)
    \right],
    \label{eq:proximal-gradient-residual}
\end{equation}
where \(\sigma\ge\lambda_{\max}(\nabla^2 q(\bm{x}))\). The nonexpansiveness of the proximal mapping and the continuity of \(\bm{g}\) imply that \(\bm{G}_{\sigma}\) is continuous. The proximal optimality condition also shows that \(\bm{G}_{\sigma}(\bm{x})=\bm{0}\) for \(\bm{x}\in\mathcal{X}\). Moreover, its norm is bounded above by the minimum-norm residual $\|\bm{G}_{\sigma}(\bm{x})\|
\leq
\|\bm{r}(\bm{x})\|
=
\operatorname{dist}
\bigl(
    \bm{0},
    \partial F(\bm{x})
\bigr)$.

We next isolate the active polyhedral structure at a stationary point. For $\bm{x}^\star \in \mathcal{X}$, define the free and active index sets 
\begin{equation}
S^{\star}:=\left\{i: x_i^{\star} \notin \mathcal{B}^\star\right\}, \quad Z^{\star}:=[n] \backslash S^{\star} .
\end{equation}
An equivalent definition of the free set is $S^{\star}:=\left\{i: h'_{i,-}\left(x_{ i}^\star\right)=h'_{i,+}\left(x_{ i}^\star\right)\right\}$, since convexity gives $h'_{i,-}\left(x_{ i}^\star\right)\leq h'_{i,+}\left(x_{ i}^\star\right)$. We call $x^{\star}$ nondegenerate, or strictly complementary, if
\begin{equation}
-g_i\left(x^{\star}\right) \in \operatorname{ri} \partial h_i\left(x_i^{\star}\right), \quad i \in Z^{\star} .
\end{equation}
Strict complementarity provides a positive separation from both boundary equalities in the coordinate-wise stationarity conditions. Together with convergence and a vanishing optimality residual, this separation is the mechanism used in finite-identification analyses of proximal and coordinate methods; see, e.g., \cite{bareilles2023newton,liang2017activity}. Based on the design of the proposed method, we later establish its finite-identification property.

\section{Algorithm Construction}
This section develops a proximal line search method for joint active-set identification and Newton-type acceleration. At \(\bm{x}_k\), Section \ref{subsec:III-A} predicts a working set $W_k$ by augmenting the standard first-order selection with small-margin kink coordinates to exploit Hessian coupling. Section \ref{subsec:III-B} computes $\bm{p}_k$ on this set by solving the quadratic model inexactly but enforcing its box constraints exactly, allowing coordinates to reach new kinks within selected affine intervals. A relative KKT-residual criterion controls the inexactness. 
To retain the nonsmooth structure revealed at \(\bm{x}_k+\bm{p}_k\) during globalization, Section \ref{subsec:III-C} uses a subgradient at this endpoint to choose the linear term of an isotropic proximal model whose initial trial recovers this endpoint exactly. Backtracking increases the proximal curvature until sufficient decrease holds, generating trial points through separable proximal evaluations. Unlike uniform scaling of $\bm{p}_k$, this search can retain endpoint kink coordinates while damping other components. The inexactness criterion and proximal construction jointly yield a composite first-order descent bound along the search path. Section \ref{subsec:property} analyzes properties of the algorithm. 


\subsection{Active-set prediction with curvature coupling}
\label{subsec:III-A}
At the current iterate \(\bm{x}_k\), we classify each coordinate according to whether \(h_i\) is differentiable at \(x_{k,i}\). Define
\begin{equation}
    S_k\coloneqq\set{i:h'_{i,-}\left(x_{k, i}\right)=h'_{i,+}\left(x_{k, i}\right)},~
    Z_k\coloneqq[n]\setminus S_k.
   \nonumber
\end{equation}
Thus, \(S_k\) and \(Z_k\) partition \([n]\) into smooth and nonsmooth coordinates, respectively. This partition serves as the starting point for active-set prediction. We construct a working set of coordinates that may contribute to descent and define the predicted active set as its complement. A standard first-order rule places every \(i\in S_k\) in the working set, together with any \(i\in Z_k\) that violates the coordinate-wise stationarity condition. We additionally include selected small-margin indices from \(Z_k\), since off-diagonal Hessian coupling with other coordinates may still make their movement beneficial.\footnote{The formal quantitative analysis is given in Section \ref{subsec:property}.} Specifically, for each \(i\in Z_k\), we assess coordinate-wise stationarity by computing
\begin{equation}
\begin{aligned}
    \xi_{k,i} &= \arg \min_{\xi \in \{-1,+1\}} F^{\prime}\left(\bm{x}_k ; \xi \bm{e}_i\right)\\
    \delta_{k,i} &= \min_{\xi \in \{-1,+1\}} F^{\prime}\left(\bm{x}_k ; \xi \bm{e}_i\right).\nonumber
    \end{aligned}
\end{equation}
Here, \(\xi_{k,i}\bm{e}_i\) is a minimizing direction within the coordinate unit ball, and \(\delta_{k,i}\) is the corresponding minimum directional derivative. Thus, \(\delta_{k,i}<0\) indicates first-order descent along \(i\), whereas \(\delta_{k,i}=0\) indicates the stationarity. A positive \(\delta_{k,i}\) means that moving coordinate \(i\) alone in either signed direction increases the objective. Nevertheless, when \(\delta_{k,i}\) is small, a joint step involving \(i\) and other coordinates may still decrease the objective because of off-diagonal Hessian coupling. Define
\begin{equation}
    V_k:=\left\{i \in Z_k: \delta_{k, i}<0\right\}, ~ E_k:=\left\{i \in Z_k \backslash V_k: 0 \leq \delta_{k, i}<\nu_k\right\},\nonumber
\end{equation}
where $\nu_k>0$. The released kink, working, and active sets are
\begin{equation}
    J_k:=V_k \cup E_k,~ W_k:=S_k \cup J_k, ~ N_k:=W_k^c .
    \label{eq:set}
\end{equation}
\subsection{Inexact Newton direction from constrained QP}
\label{subsec:III-B}
We jointly consider the local polyhedral geometry of the nonsmooth term with the curvature of the smooth term to obtain the Newton-type direction. It enforces structural feasibility exactly while solving the resulting quadratic model inexactly. Specifically, given \(W_k\) and the selected signs \(\{\xi_{k,i}\}_{i\in J_k}\), let \(\bm p\) denote a Newton-type direction. For each \(i\in W_k\), we constrain \(x_{k,i}+p_i\) to the corresponding affine interval of \(h_i\). This restriction makes \(h(\bm x_k+\bm p)-h(\bm x_k)\) exactly linear in \(\bm p_{W_k}\). Consequently, computing \(\bm p\) reduces to approximately solving a strongly convex, box-constrained quadratic program subject to \(\bm p_{N_k}=0\). We first define the affine intervals. For $i \in W_k$, let $b_{k, i}^{-}$and $b_{k, i}^{+}$denote the nearest breakpoints of $h_i$ strictly to the left and right of $x_{k, i}$, respectively. That is $b_{k, i}^{-}:= \max \left\{b \in \mathcal{B}: b<x_{k, i}\right\},~b_{k, i}^{+}:=  \min \left\{b \in \mathcal{B}: b>x_{k, i}\right\}$. The affine interval is then 
\begin{equation}
    I_{k, i}:= \begin{cases}{[b_{k, i}^{-}, b_{k, i}^{+}],} & i \in S_k, \\ {[x_{k, i}, b_{k, i}^{+}],} & i \in J_k, \xi_{k, i}=+1, \\ {[b_{k, i}^{-}, x_{k, i}],} & i \in J_k, \xi_{k, i}=-1 .\end{cases}
    \label{eq:interval}
\end{equation}
Write $I_{k, i}=[l_{k, i}, u_{k, i}]$ for simplicity, and let $\alpha_{k, i}$ denote the slope of $h_i$ on this interval; explicitly,
\begin{equation}
\alpha_{k, i}:= \begin{cases}h_i^{\prime}\left(x_{k, i}\right), & i \in S_k, \\ h_{i, \xi_{k, i}}^{\prime}\left(x_{k, i}\right), & i \in J_k .\end{cases}
\label{eq:slope}
\end{equation}
Consequently, $h_i(t)-h_i\left(x_{k, i}\right)=\alpha_{k, i}\left(t-x_{k, i}\right),~ t \in I_{k, i}$. Because \(h(\bm{x}_k+\bm{p})-h(\bm{x}_k)\) is linear in \(\bm{p}_{W_k}\) over the interval \eqref{eq:interval}, we compute the exact Newton-type direction as
\begin{equation}
\bm{p}_k^{\star}={\arg \min_{\bm{p} \in C_k} }\left\langle\bm{g}_{k, W_k}+\bm{\alpha}_{k, W_k}, \bm{p}_{W_k}\right\rangle+\frac{1}{2} \bm{p}_{W_k}^{\top}\left(\bm{B}_k\right)_{W_k W_k} \bm{p}_{W_k},
\label{eq:exact-Newton-direction}
\end{equation}
where $C_k:=\left\{\bm{p} \in \mathbb{R}^n: \bm{p}_{N_k}=\bm{0},~ x_{k, i}+p_i \in I_{k, i}~ \forall i \in W_k\right\}$, and $\bm{B}_k$ is an appropriate approximation to $\nabla^2 q(\bm{x}_k)$ satisfying 
\begin{equation}
    m \bm I \preceq\left(\bm B_k\right)_{W_k W_k} \preceq M \bm I
    \label{eq:bounded-hessian}
\end{equation}
for some $0<m<M<\infty$. We next introduce an inexactness criterion for \eqref{eq:exact-Newton-direction}. Let $C_{k, W_k}$ be the reduced feasible set in $\mathbb{R}^{\left|W_k\right|}$ and define $\bm{a}_k:=\bm{g}_{k, W_k}+\bm{\alpha}_{k, W_k},~ \bm{B}_k^W:=\left(\bm{B}_k\right)_{W_k W_k}, ~\bm{q}_k(\bm{p}):=\bm{a}_k+\bm{B}_k^W \bm{p}_{W_k}$ for simplicity. The KKT residual to \eqref{eq:exact-Newton-direction} is
\begin{equation}
\varepsilon_k(\bm{p}):=\operatorname{dist}_{\left(\bm{B}_k^W\right)^{-1}}\left(\bm{0}, \bm{a}_k+\bm{B}_k^W \bm{p}_{W_k}+N_{C_{k, W_k}}\left(\bm{p}_{W_k}\right)\right).\nonumber
\end{equation}
Each iteration, we solve \eqref{eq:exact-Newton-direction} for $\bm{p}_k \in C_k$, and 
\begin{equation}
\varepsilon_k\left(\bm{p}_k\right) \leq \vartheta_k\left\|\bm{p}_{k, W_k}\right\|_{\bm{B}_k^W}, \quad 0 \leq \vartheta_k \leq \bar{\vartheta}<\frac{1}{2} .
\label{eq:inexact-condition}
\tag{\textcolor{red}{IC}}
\end{equation}
Accordingly, there exist $\bm{n}_k$ $\in N_{C_{k, W_k}}\left(\bm{p}_{k, W_k}\right)$ such that the residual $\bm{e}_k=\bm{q}_k(\bm{p}_{k,W_k})+\bm{n}_k$, satisfies $\left\|\bm{e}_k\right\|_{\left(\bm{B}_k^W\right)^{-1}}=\varepsilon_k\left(\bm{p}_k\right)$. The bounded $\bm{B}_k^W$ gives $\left\|\bm{e}_k\right\| \leq \sqrt{M}\left\|\bm{e}_k\right\|_{\left(\bm{B}_k^W\right)^{-1}} \leq M \vartheta_k\left\|\bm{p}_k\right\|$. To avoid computing a weighted projection onto the model normal cone, we use the following inexpensive sufficient condition. For a feasible $\bm{p}_k$, $\bm{E}_k \in \mathbb{R}^{\left|W_k\right|}$ componentwise by
\begin{equation}
\left[\bm{E}_k\right]_i:= \begin{cases}{\left[\bm{q}_k(\bm{p}_k)\right]_i,} & l_{k,i}<p_{k,i}+x_{k,i}<u_{k,i}, \\ \min \left\{\left[\bm{q}_k(\bm{p}_k)\right]_i, 0\right\}, & p_{k,i}+x_{k,i}=l_{k,i}<u_{k,i}, \\ \max \left\{\left[\bm{q}_k(\bm{p}_k)\right]_i, 0\right\}, & p_{k,i}+x_{k,i}=u_{k,i}>l_{k,i}, \\ 0, & l_{k, i}=u_{k, i} .\end{cases}
\label{eq:inexact-condition-2}
\end{equation}
We show the connection between \eqref{eq:inexact-condition} and \eqref{eq:inexact-condition-2}.
\begin{lemma}
    For a feasible $\bm{p}_k$, 
    \begin{equation}
\left\|\bm{E}_k\right\|=\operatorname{dist}\left(\bm{0}, \bm{q}_k(\bm{p}_k)+N_{C_{k, W_k}}\left(\bm{p}_{k,W_k}\right)\right) .
\label{eq:inexact-condition-3}
\end{equation}
Consequently, $\left\|\bm{E}_k\right\| \leq m \vartheta_k\left\|\bm{p}_{k, W_k}\right\|$ implies \eqref{eq:inexact-condition}.\label{lem:1}
\end{lemma}
\begin{proof}
    See Appendix \ref{appendix-A}.
\end{proof}
Because $\bm{B}_k^W\succeq m\bm{I}$ and $C_{k,W_k}$ is a box, the reduced model is a strongly convex box-constrained quadratic program.  Hence the inexact condition is attainable by standard feasible set, projected-gradient, or coordinate methods; see, e.g., \cite{byrd2016inexact,lee2019inexact,kanzow2021globalized}.

\subsection{Globalization via proximal line search}
\label{subsec:III-C}
The proposed globalization strategy aims to retain the nonsmooth pattern revealed by the Newton-type direction while ensuring sufficient decrease. A conventional line search uniformly damps the step and may therefore move coordinates away from kinks reached at \(\bm{x}_k+\bm{p}_k\). We instead embed \(\bm{p}_k\) into a specially constructed isotropic proximal surrogate  constructed so that its initial trial returns \(\bm{x}_k+\bm{p}_k\) exactly. As backtracking increases the proximal curvature, coordinates that have reached kinks are encouraged to remain fixed there, while the remaining coordinates are gradually damped. For $\bm{p}_k \neq 0$, define $\beta_k
    \coloneqq
    \frac{\bm{p}_k^{\mathsf T}\bm{B}_k\bm{p}_k}{\norm{\bm{p}_k}^2},~
    m\le\beta_k\le M$ measures the normalized curvature of the quadratic model along $\bm{p}_k$. We choose $\boldsymbol{v}_k \in \partial h\left(\bm{x}_k+\bm{p}_k\right)$ to minimize the endpoint stationarity residual
\begin{equation}
    {\bm{v}}_k
    \coloneqq
    \proj_{\partial h(\bm{x}_k+\bm{p}_k)}
    \paren{-\bm{g}_k-\beta_k\bm{p}_k},
    \label{eq:vbar}
\end{equation}
which is componentwise projection onto the intervals in \eqref{eq:interval}. Define
\begin{equation}
    \bm{c}_k
    \coloneqq
    -\bm{g}_k-\beta_k\bm{p}_k-{\bm{v}}_k.
    \label{eq:correction}
\end{equation}
For any \(\lambda \ge \beta_k\), the isotropic proximal surrogate used in backtracking is given by 
\begin{equation}
    \begin{aligned}
       & \bm{z}_k(\lambda)
   =
    \prox_{\lambda^{-1}h}
    \paren{
        \bm{x}_k-\lambda^{-1}(\bm{g}_k+\bm{c}_k)
    }\\
    =&\arg \min_{\bm{z}} ~h(\bm z)+\left\langle \bm g_k+\bm c_k, \bm z-\bm x_k\right\rangle+\frac{\lambda}{2}\left\|\bm z-\bm x_k\right\|^2 .
    \end{aligned}
    \label{eq:main-step}
\end{equation}
The first-order optimality condition for \(\bm{z}_k(\lambda)\) is
\begin{equation}
0\in
\bm{g}_k+\bm{c}_k
+\lambda\left(\bm{z}_k(\lambda)-\bm{x}_k\right)
+\partial h\left(\bm{z}_k(\lambda)\right).
\end{equation}
\begin{remark}
    At \(\lambda=\beta_k\), consider the candidate \(\bm{z}=\bm{x}_k+\bm{p}_k\). The corresponding condition becomes $\bm{g}_k+\bm{c}_k+\beta_k\bm{p}_k
+\partial h(\bm{x}_k+\bm{p}_k)$. By construction, $\bm{g}_k+\bm{c}_k+\beta_k\bm{p}_k+\bm{v}_k=0,
~
\bm{v}_k\in\partial h(\bm{x}_k+\bm{p}_k)$. Hence, \(\bm{x}_k+\bm{p}_k\) satisfies the optimality condition. Since the proximal subproblem has a unique minimizer, $\bm{z}_k(\beta_k)=\bm{x}_k+\bm{p}_k$ and the full step is recovered. 
\end{remark}
At each iteration, backtracking starts from \(\lambda=\beta_k\) and increases \(\lambda\) geometrically, moving \(\bm{z}_k(\lambda)\) toward \(\bm{x}_k\). For any fixed \(\lambda\), Theorem~\ref{the:1} gives a closed form for \eqref{eq:main-step} and reveals non-uniform damping that tends to preserve the nonsmooth pattern at the endpoint. Define
\begin{equation}
\begin{aligned}
D_k := {}&
\{i\in W_k:p_{k,i}>0,\ x_{k,i}+p_{k,i}=u_{k,i}\}\\
&{}\cup
\{i\in W_k:p_{k,i}<0,\ x_{k,i}+p_{k,i}=l_{k,i}\},
\end{aligned}
\end{equation}
and the composite first-order change $
\Delta_k(\bm d)
:=
\langle\bm g_k,\bm d\rangle
+h(\bm x_k+\bm d)-h(\bm x_k),$ and set $
\sigma(\vartheta)
:=
1-\frac{1}{4(1-\vartheta)}, ~\bar \sigma = \sigma(\bar \vartheta)>1/2$, where $\bar\vartheta<1/2$ is the uniform bound in \eqref{eq:inexact-condition}. Set $\bm d_k(\lambda)
:=
\bm z_k(\lambda)-\bm x_k,~
\omega_{k,i}:=\operatorname{sign}(p_{k,i})(v_{k,i}-\alpha_{k,i})\ge 0$.

\begin{algorithm}[t] \caption{Inexact Proximal Newton Method} \label{alg:ipqn} \begin{algorithmic}[1] \Require $\boldsymbol{x}_0$, $\gamma>1$, $\eta\in(0,1)$,
$\{\nu_k\}_{k\geq0}$ with $\nu_k>0$,
$\{\vartheta_k\}_{k\geq0}$ with
$0\leq\vartheta_k\leq\bar{\vartheta}<1/2$, and $\epsilon>0$
\For{$k=0,1,\ldots$} 
\State Evaluate $\bm g_k$ and $\bm r(\bm x_k)$
\State \textbf{if} $\|\bm{r}(\bm{x}_k)\|\leq\epsilon$, \textbf{return} $\bm{x}_k$

\State Compute $W_k$, $N_k$, and $\{I_{k,i},\alpha_{k,i}\}_{i\in W_k}$ via \eqref{eq:set}, \eqref{eq:interval}, and \eqref{eq:slope}

\State Construct $\bm{B}_k$ satisfying \eqref{eq:bounded-hessian}

\State Compute $\bm{p}_k\in C_k$ satisfying \eqref{eq:inexact-condition}

\State Backtrack to the first trial $\lambda_{k,j_k}$ satisfying
\eqref{eq:armijo}

\State $\bm{x}_{k+1}\gets \bm{z}_k(\lambda_{k,j_k})$ \EndFor \end{algorithmic} \end{algorithm}

\begin{theorem}
    \label{the:1}
    Let $\bm p_k\in C_k\setminus\{\bm0\}$ satisfy \eqref{eq:inexact-condition}, and
$\lambda\ge\beta_k$, then, for every $i\in[n]$
\begin{equation}
    d_{k,i}(\lambda)
=
\begin{cases}
\dfrac{\beta_k}{\lambda}p_{k,i},
    & i\notin D_k,\\[1.2ex]
\min\!\left\{
    1,\dfrac{\beta_k+{\omega_{k,i}}/{|p_{k,i}|}}{\lambda}
    \right\}p_{k,i},
    & i\in D_k.
\end{cases}
\label{eq:the:1-1}
\end{equation}
Consequently, $\frac{\beta_k}{\lambda}\|\bm p_k\|
\le
\|\bm d_k(\lambda)\|
\le
\|\bm p_k\|$. For every $i\in D_k$, $z_{k,i}(\lambda)=y_{k,i}$ exactly when $\lambda
    \le
    \beta_k+{\omega_{k,i}}/{\abs{p_{k,i}}}$. Finally, the entire proximal path satisfies
    \begin{equation}
        \Delta_k\bigl(\bm d_k(\lambda)\bigr)
\le
-\bar\sigma\lambda
\|\bm d_k(\lambda)\|^2.
\label{eq:the:1-2}
    \end{equation}
\end{theorem}
\begin{proof}
    See Appendix \ref{appendix-B}.
\end{proof}
The backtracking stops when a sufficient descent condition is satisfied. Given \(\gamma>1\) and \(\eta\in(0,1)\), backtracking tests \(\lambda_{k,j}=\beta_k\gamma^j\) for \(j=0,1,\ldots\) and accepts the first value satisfying
\begin{equation}
F\bigl(\bm{z}_k(\lambda_{k,j})\bigr)
\le
F(\bm{x}_k)
-\eta\frac{\lambda_{k,j}}{2}
\norm{\bm{d}_k(\lambda_{k,j})}^2.
\label{eq:armijo}
\end{equation}
For the accepted \(\lambda_{k,j}\), set $\bm{x}_{k+1} = \bm{z}_k(\lambda_{k,j})$. Algorithm~\ref{alg:ipqn} summarizes the complete method.

\subsection{Algorithm properties}
\label{subsec:property}
This section derives error and descent bounds under the inexact condition, shows the benefit of using curvature information to expand the working set, and proves finite termination of proximal backtracking. For $\bm{p} \in C_k$, define the local quadratic model $m_k(\bm{p}):=\left\langle \bm{a}_k, \bm{p}_{W_k}\right\rangle+\frac{1}{2}\left\|\bm{p}_{W_k}\right\|_{\bm{B}_k^W}^2.$
\begin{lemma}
\label{lem:inexact}
Suppose \eqref{eq:bounded-hessian} holds.  Let $\bm{p}_k \in C_k$ satisfy \eqref{eq:inexact-condition} with the residual $\bm{e}_k$. Then $\bm{p}_k^\star$ is unique and
\begin{align}
&\left\|\bm{p}_{k, W_k}-\bm{p}_{k, W_k}^{\star}\right\|_{\bm{B}_k^W}
\leq \left\|\bm{e}_k\right\|_{\left(\bm{B}_k^W\right)^{-1}} \label{eq:lem-2-1}\\
&0 \leq m_k\left(\bm{p}_k\right)-m_k\left(\bm{p}_k^{\star}\right)
\leq \frac{1}{2}\left\|\bm{e}_k\right\|_{\left(\bm{B}_k^W\right)^{-1}}^2\label{eq:lem-2-2}.
\end{align}
Moreover, $\bm{p}_k = 0$ if and only if $\bm{r}(\bm{x}_k)=0$.
\end{lemma}
\begin{proof}
    See Appendix \ref{appendix-C}.
\end{proof}
We next clarify the role of the small-margin set \(E_k\). The following two propositions show that its coordinates become nonzero only under sufficiently negative coupling and quantify the resulting model decrease, accounting for inexactness.
\begin{proposition}
\label{prop:2}
    Let $\bm{p}_k^\star$ be the exact minimizer of \eqref{eq:exact-Newton-direction}. For $j\in E_k$, $p^\star_{k,j} \neq 0$ if and only if 
    \begin{equation}
-\xi_{k, j} \sum_{i \in W_k \backslash\{j\}}\left(\bm{B}_k\right)_{j i} p_{k, i}^{\star}>\delta_{k, j} .
\label{eq:prop-2-1}
\end{equation}
\end{proposition}
\begin{proof}
    See Appendix \ref{appendix-D}.
\end{proof}

\begin{proposition}
\label{prop:3}
    Let $\widehat{\bm{p}}_k$ be the exact solution of the model where $E_k = \varnothing$, $\widehat{\bm{p}}_k=\underset{\bm{p} \in \mathcal{C}_k, \bm{p}_{E_k}=\bm{0}}{\operatorname{argmin}} m_k(\bm{p})$. Let $\bm{p}_k^\star$ be the exact minimizer of \eqref{eq:exact-Newton-direction}. Then
    \begin{equation}
m_k\left(\widehat{\bm{p}}_k\right)-m_k\left(\bm{p}_k^{\star}\right) \geq \frac{1}{2}\left\|\widehat{\bm{p}}_{k, W_k}-\bm{p}_{k, W_k}^{\star}\right\|_{\bm{B}_k^W}^2 \geq 0 .
\end{equation}
Moreover $m_k\left(\bm{p}_k^{\star}\right)<m_k\left(\widehat{\bm{p}}_k\right)$ if and only if $\bm{p}_{k, E_k}^{\star} \neq \bm{0}$. For the computed inexact step $\bm{p}_k$, it further gives
\begin{equation}
\begin{aligned}
m_k\left(\widehat{\bm{p}}_k\right)-m_k\left(\bm{p}_k\right) \geq \frac{1}{2}&\left(\left\|\widehat{\bm{p}}_{k, W_k}-\bm{p}_{k, W_k}^{\star}\right\|_{\bm{B}_k^W}^2\right. \\
&\left.-\left\|\bm{e}_k\right\|_{\left(\bm{B}_k^W\right)^{-1}}^2\right) .
\end{aligned}
\label{eq:prop-3-1}
\end{equation}
\end{proposition}
\begin{proof}
    See Appendix \ref{appendix-E}.
\end{proof}
We show the backtracking is well-defined.
\begin{theorem}
    \label{the:2}
    Suppose Assumption \ref{ass:main-1} and condition \eqref{eq:bounded-hessian} hold and let $\bm{p}_k$ satisfy \eqref{eq:inexact-condition}, then every $\lambda_{k, j} \geq \frac{L}{2 \bar{\sigma}-\eta}$ satisfies the sufficient decrease condition \eqref{eq:armijo}. Moreover, if $\lambda_k$ is the accepted value, then $\beta_k \leq \lambda_k \leq \bar{\lambda},$ where $\bar{\lambda}=\max \left\{M, \frac{\gamma L}{2 \bar{\sigma}-\eta}\right\} .$ Consequently, $\frac{\beta_k}{\lambda_k} \geq \frac{m}{\bar{\lambda}}>0 .$
\end{theorem}
\begin{proof}
    See Appendix \ref{appendix-F}.
\end{proof}

\section{Convergence Analysis}
This section analyzes convergence to stationary points, finite identification, and the local convergence rate of the proposed algorithm. We first establish a uniform descent bound for accepted steps and show that every accumulation point is stationary. The main difficulty is that a coordinate may approach a kink from either side.
\begin{theorem}
\label{thm:3}
Suppose assumptions of Theorem \ref{the:2} hold. Then
\begin{equation}
F(\bm x_k)-F(\bm x_{k+1})
\ge
\frac{\eta m^2}{2\bar\lambda}\|\bm p_k\|^2.
\label{eq:global-descent}
\end{equation}
Consequently, $\sum_{k=0}^{\infty}\left\|\bm{p}_k\right\|^2<\infty$ and $\lim_{k\to\infty}\|\bm{p}_k\|
=
\lim_{k\to\infty}\|\bm{d}_k\|
=
\lim_{k\to\infty}\|\bm{e}_k\|
=0.$ Moreover, every accumulation point of $\{\bm x_k\}$ is stationary.
In particular,
\begin{equation}
\operatorname{dist}(\bm x_k,\mathcal X)\to0.
\label{eq:distance-stationary}
\end{equation}
\end{theorem}
\begin{proof}
    See Appendix \ref{appendix-G}.
\end{proof}
Theorem \ref{thm:3} ensures that every accumulation point is stationary, but it does not yet guarantee convergence to a single stationary point or finite stabilization of the polyhedral structure.
\subsection{Finite identification}
The identification analysis has two steps. First, we use strict complementarity and reduced-Hessian non-singularity to isolate \(\bm x^\star\), which upgrades sub-sequential convergence to convergence of the full sequence. Second, the strict-complementarity margin, together with our proximal update \eqref{eq:main-step}, forces every coordinate in \(Z^\star\) to reach its limiting kink in finitely many iterations and prevents it from being moved again. The analysis is inspired by \cite{hare2004identifying}. We first recall some basic definitions and introduce the assumptions. Let \(\bm x^\star\) be an accumulation point of \({\bm x_k}\). Then \(\bm x^\star\in\mathcal X\) by Theorem \ref{thm:3}. Recall the free and active index sets $i^\star\notin\mathcal B,~
Z^\star=[n]\setminus S^\star.$ For each \(i\in S^\star\), let \(I_i^\star\) denote the open affine interval of \(h_i\) containing \(x_i^\star\). For the limiting kink coordinates, define the strict-complementarity margin
\begin{equation}
\Delta^{\star}:= \begin{cases}\min _{i \in Z^{\star}}\left\{F^{\prime}\left(\bm{x}^{\star} ; \bm{e}_i\right), F^{\prime}\left(\bm{x}^{\star} ;-\bm{e}_i\right)\right\}, & Z^{\star} \neq \varnothing, \\ +\infty, & Z^{\star}=\varnothing .\end{cases}
\nonumber
\end{equation}
Since \(\bm x^\star\) is stationary, \(\Delta^\star\ge0\).
\begin{assumption}
    \label{assump-2}
    Let \(\bm x^\star\) be an accumulation point of \({\bm x_k}\) and $q$ be twice continuously differentiable on a neighborhood of \(\bm x^\star\). The following conditions hold
    \begin{enumerate}[label=(\roman*)]
    \item $\Delta^\star>0,~\text{and }
\limsup_{k\to\infty}\nu_k<\Delta^\star.$
\item The reduced Hessian satisfies $\ker\paren{\bm{H}_{S^\star S^\star}(\bm{x}^\star)}=\set{\bm{0}}$.
    \end{enumerate}
\end{assumption}
The condition \(\Delta^\star>0\) is equivalent to strict complementarity at every kink coordinate of a selected accumulation point. This ensures a positive margin on either side of each kink. $\limsup_{k\to\infty}\nu_k<\Delta^\star$ ensures, near \(\bm x^\star\), a limiting kink coordinate cannot enter either \(V_k\) or \(E_k\). Non-singularity of the reduced Hessian locally isolates the stationary point on the free coordinates. The reduced Hessian may still be indefinite.
\begin{proposition}
    \label{prop:4}
    Suppose that Assumptions \ref{ass:main-1} and \ref{assump-2} hold. Then \(\bm{x}^\star\) is an isolated stationary point, and the entire sequence \(\{\bm{x}_k\}\) converges to \(\bm{x}^\star\).
\end{proposition}
\begin{proof}
    See Appendix \ref{appendix-H}.
\end{proof}

Define the limiting active manifold
\[\mathcal M^\star \coloneqq
\left\{
\bm x\in\mathbb R^n
\,\middle|\,
\bm x_{Z^\star}=\bm x^\star_{Z^\star},\;
x_i\in I_i^\star\ \forall i\in S^\star
\right\}.\]
The next theorem shows that the algorithm reaches this manifold in finitely
many iterations and thereafter works only on the smooth coordinates.
\begin{figure*}[!t]
    \centering

    \subfloat[\label{fig:column_a}]{%
        \begin{minipage}[b]{0.200\textwidth}
            \centering
            \includegraphics[width=\linewidth]{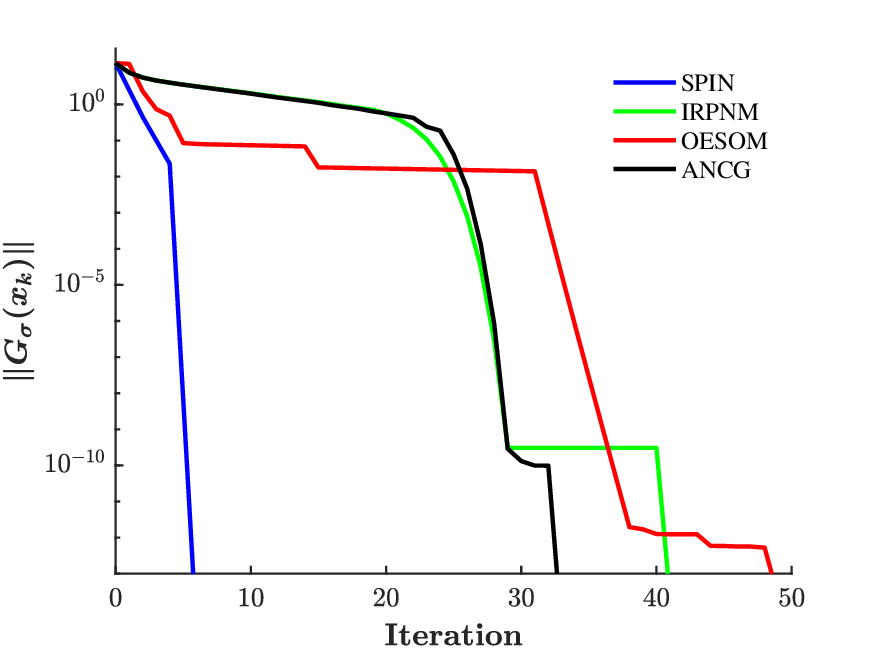}\\[-0.35mm]
            \includegraphics[width=\linewidth]{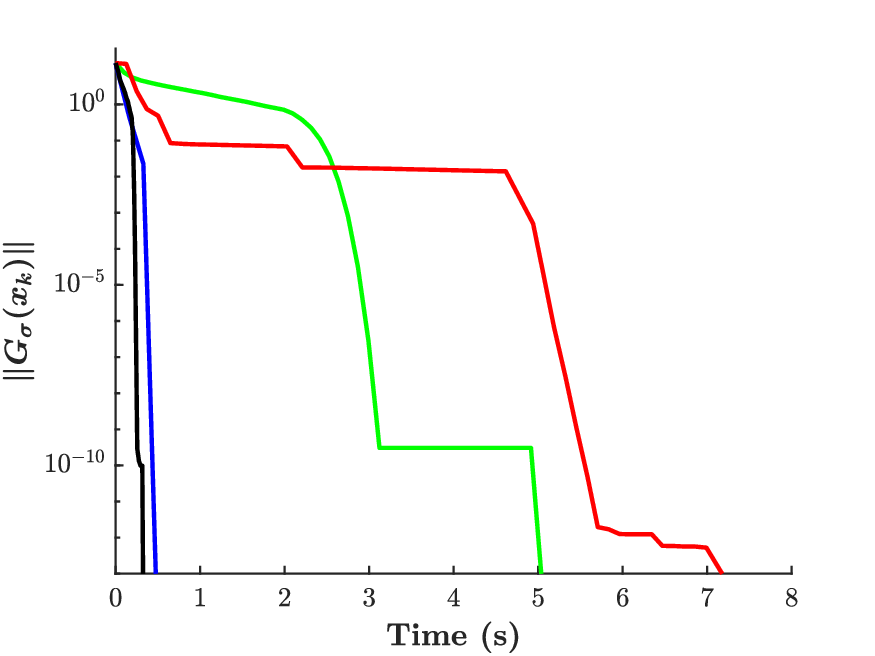}\\[-0.35mm]
            \includegraphics[width=\linewidth]{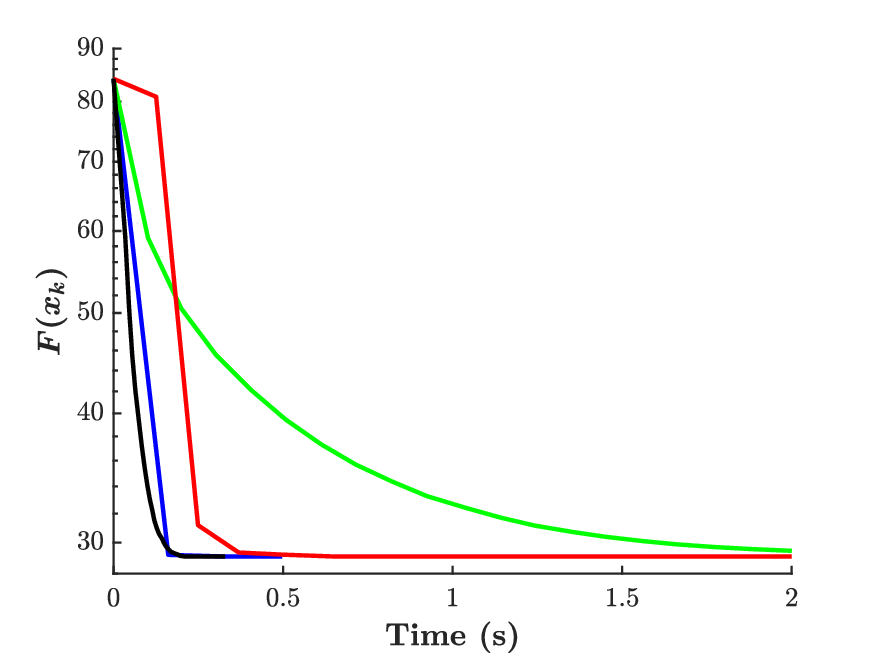}
        \end{minipage}%
    }%
    \hfill
    \subfloat[\label{fig:column_b}]{%
        \begin{minipage}[b]{0.200\textwidth}
            \centering
            \includegraphics[width=\linewidth]{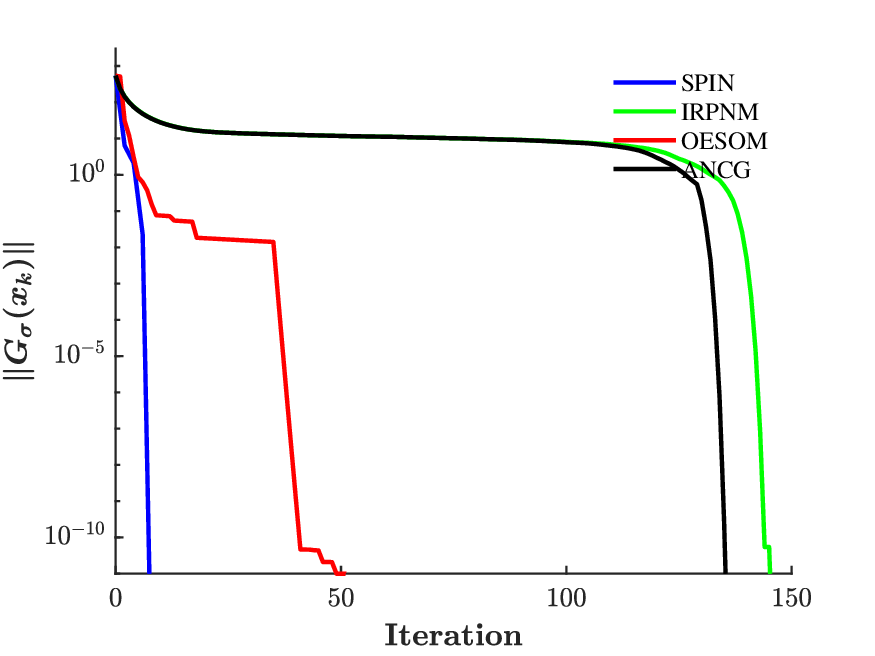}\\[-0.35mm]
            \includegraphics[width=\linewidth]{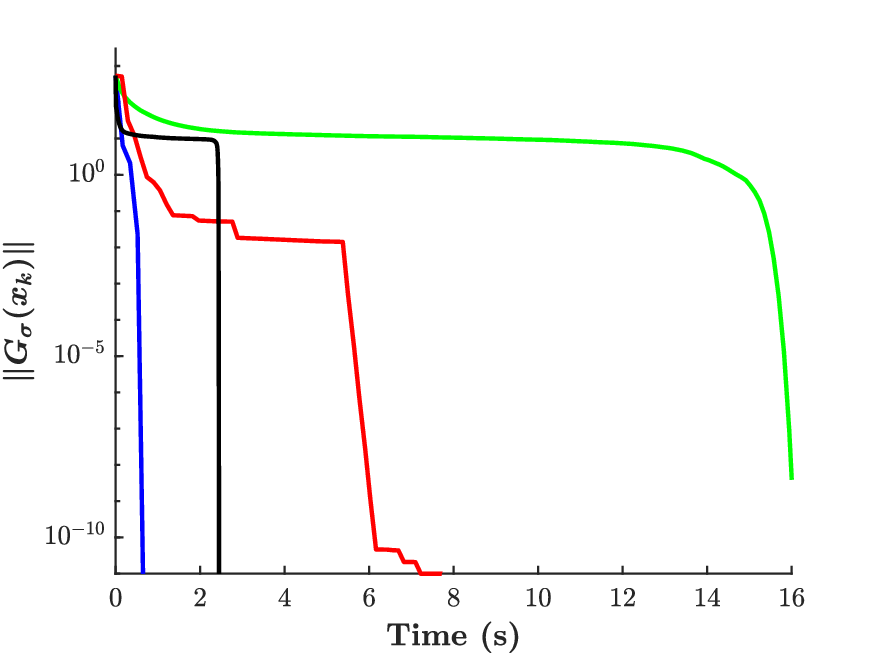}\\[-0.35mm]
            \includegraphics[width=\linewidth]{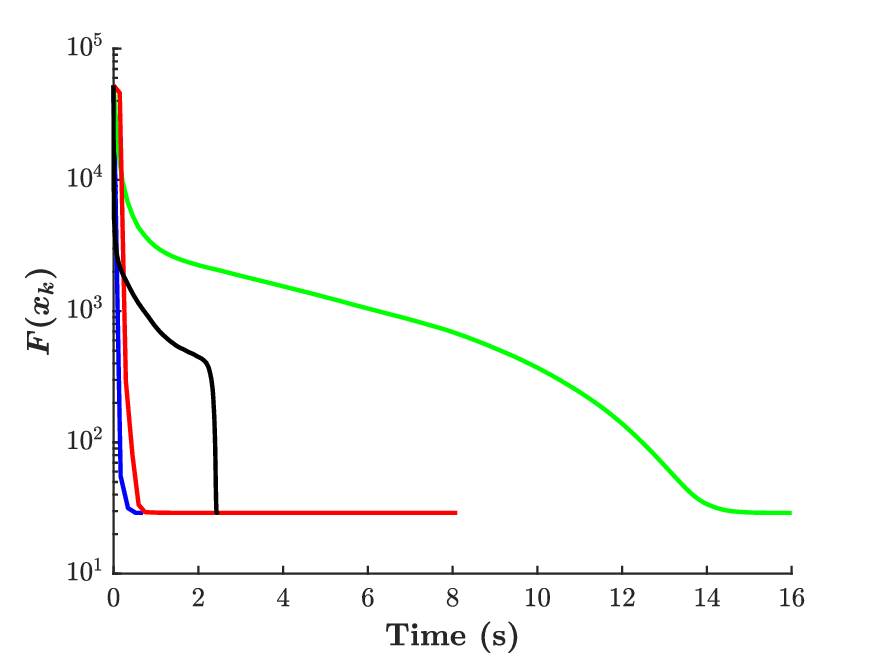}
        \end{minipage}%
    }%
    \hfill
    \subfloat[\label{fig:column_c}]{%
        \begin{minipage}[b]{0.200\textwidth}
            \centering
            \includegraphics[width=\linewidth]{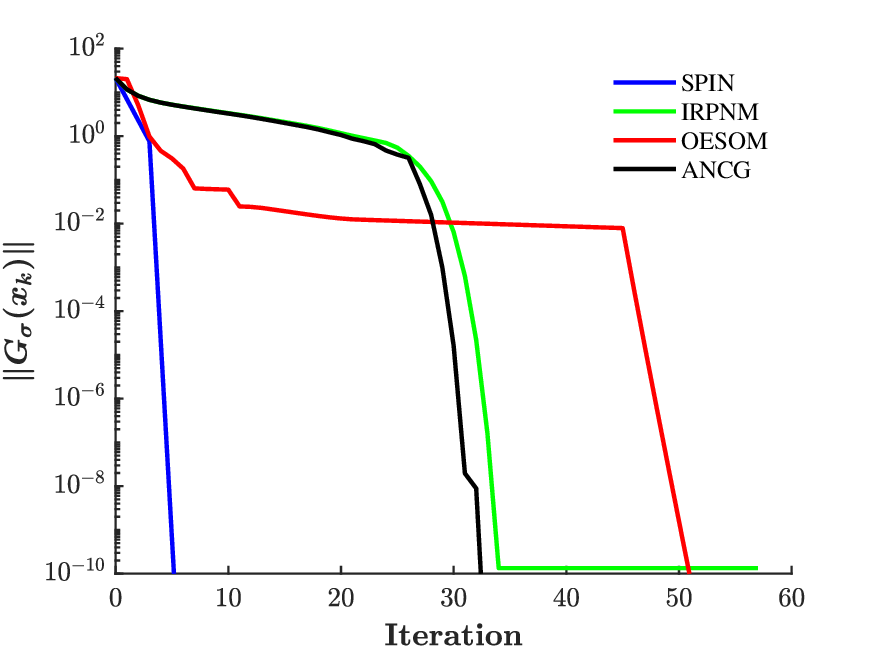}\\[-0.35mm]
            \includegraphics[width=\linewidth]{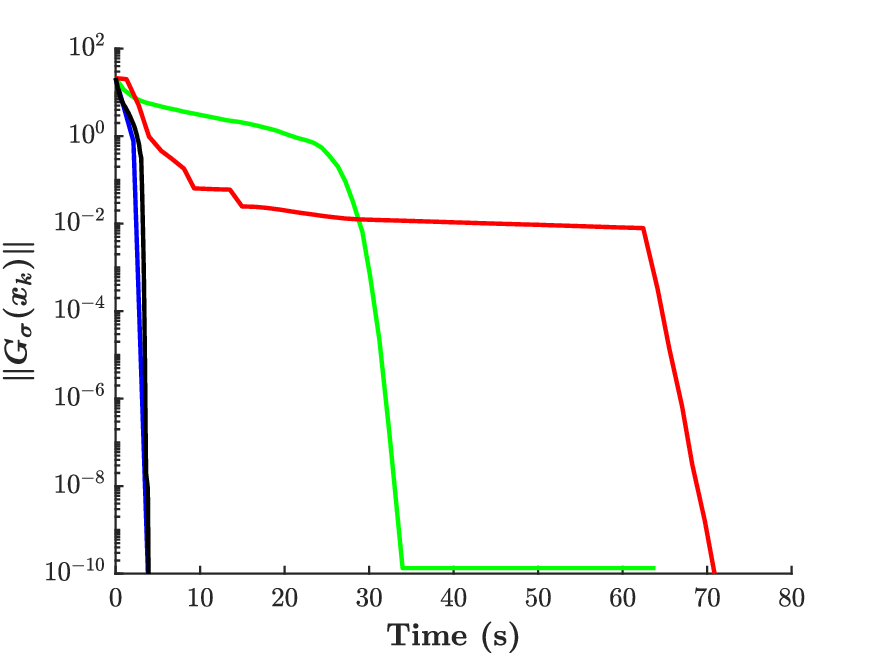}\\[-0.35mm]
            \includegraphics[width=\linewidth]{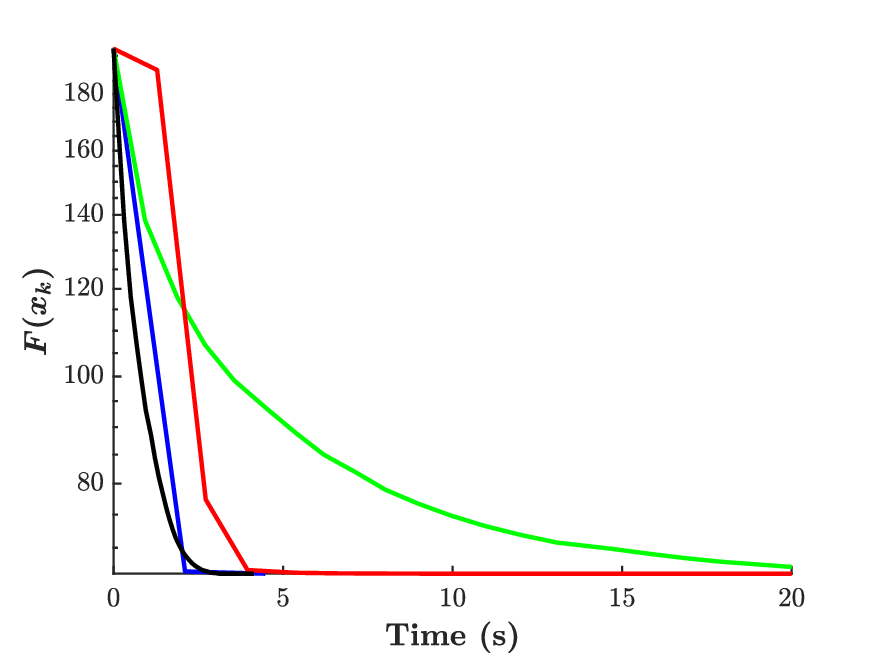}
        \end{minipage}%
    }%
    \hfill
    \subfloat[\label{fig:column_d}]{%
        \begin{minipage}[b]{0.200\textwidth}
            \centering
            \includegraphics[width=\linewidth]{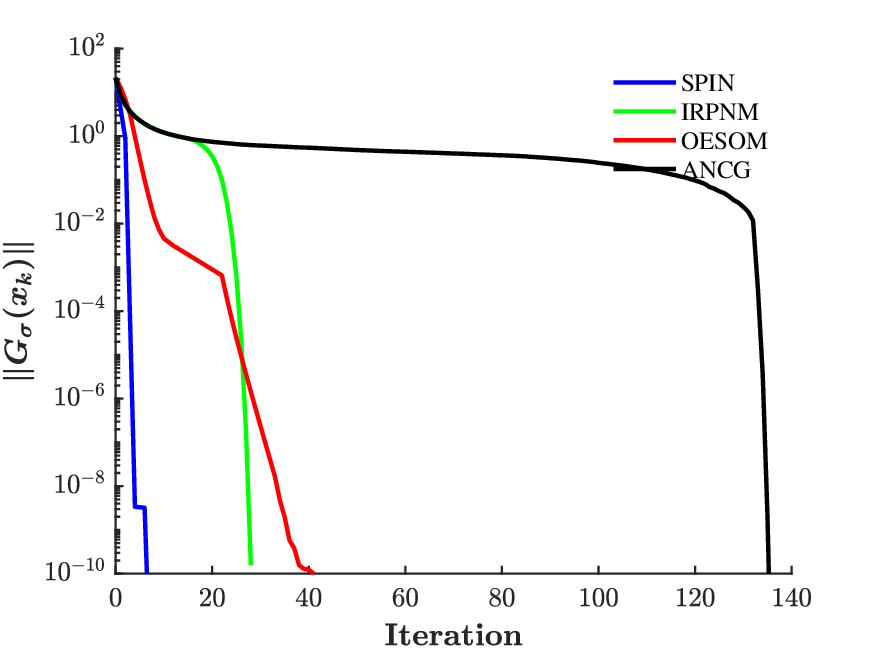}\\[-0.35mm]
            \includegraphics[width=\linewidth]{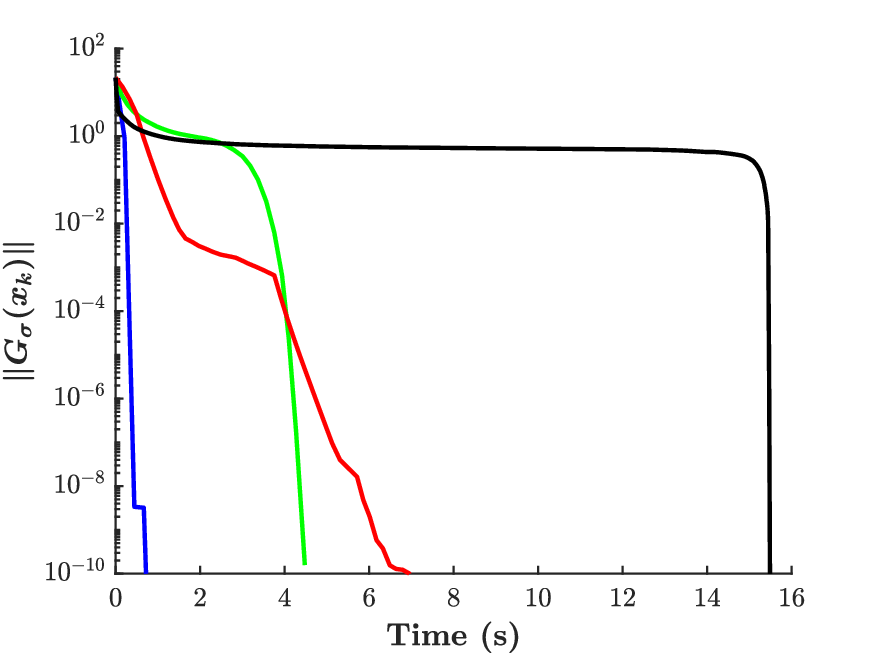}\\[-0.35mm]
            \includegraphics[width=\linewidth]{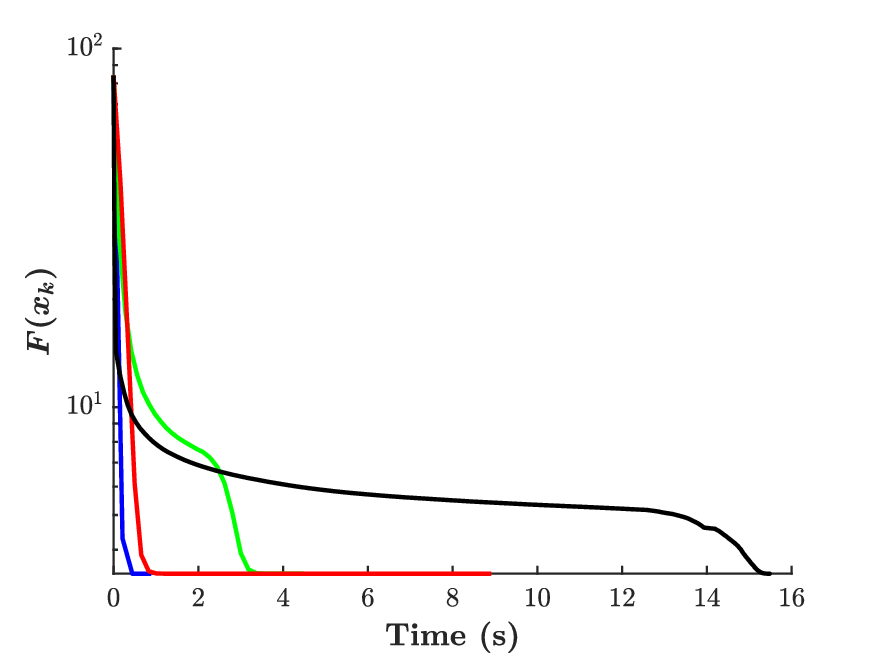}
        \end{minipage}%
    }%
    \hfill
    \subfloat[\label{fig:column_e}]{%
        \begin{minipage}[b]{0.200\textwidth}
            \centering
            \includegraphics[width=\linewidth]{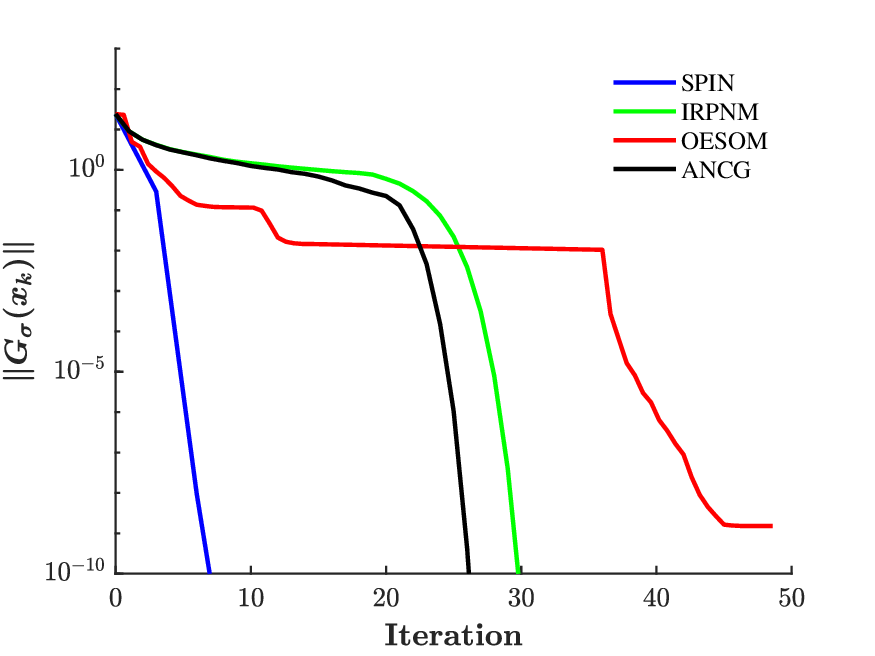}\\[-0.35mm]
            \includegraphics[width=\linewidth]{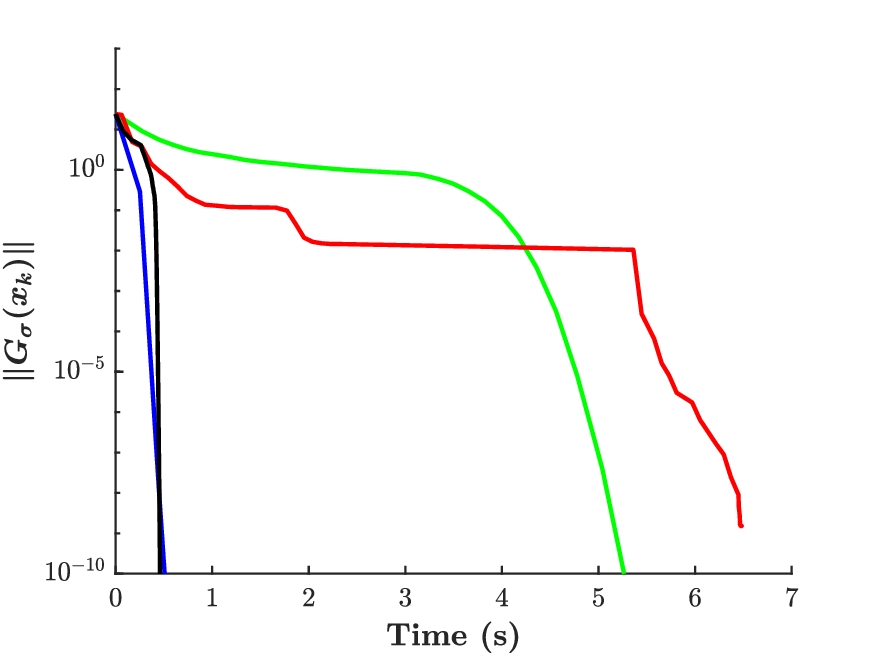}\\[-0.35mm]
            \includegraphics[width=\linewidth]{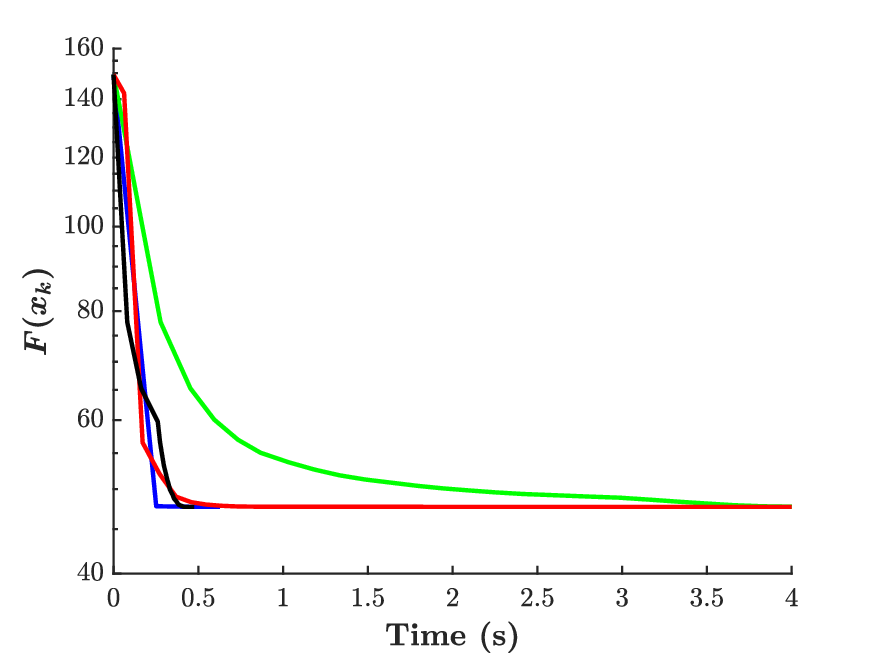}
        \end{minipage}%
    }

    \caption{Convergence trajectories for five test cases. The top and middle rows show $\|\bm{G}_{\sigma}(\bm{x}_k)\|$ versus iteration count and time in seconds, respectively, while the bottom row shows $F(\bm{x}_k)$ versus time. (a) Default setting. (b) Distant initialization with $\bm{x}_0=10 \times \texttt{rand()}$. (c) Larger problem with $m=1000$, $n=3000$, and $\kappa=100$. (d) Smaller regularization parameter with $\zeta=0.01\|\bm{A}^{\top}\bm{b}\|_{\infty}$. (e) Denser $\bm{x}^{\star}$ with $\kappa=160$.}
    \label{fig:1}
\end{figure*}
\begin{theorem}
\label{thm:4}
Suppose that Assumptions \ref{ass:main-1} and \ref{assump-2} hold. Then there exists $K_{\mathrm{id}}<\infty$ such that, for every
$k\ge K_{\mathrm{id}}$, $\bm{x}_k \in \mathcal M^\star$.
\end{theorem}
\begin{proof}
    See Appendix \ref{appendix-I}.
\end{proof}
\subsection{Local convergence rate}
Once the active manifold is identified, the method reduces to an inexact Newton scheme for the reduced smooth problem. Under a directional Dennis–Moré condition on the identified smooth coordinates and negligible inexactness relative to the step, we establish local superlinear convergence. For brevity, we use the following notation throughout this section. Write $\boldsymbol{\alpha}_{S^\star}^\star := \bigl(h_i'(x_i^\star)\bigr)_{i\in S^\star}$, and $\bm r_k^S := \bm g_{k,S^\star} + \boldsymbol{\alpha}_{S^\star}^\star,~ \bm p_k^S := \bm p_{k,S^\star}, ~ \bm e_k^S := \bm e_{k,S^\star}, $ and $\bm B_k^S := (\bm B_k)_{S^\star S^\star}, ~ \bm H_k^S := (\bm H_k)_{S^\star S^\star}, ~ \bm H_\star^S := \bm H_{S^\star S^\star}(\bm x^\star).$
\begin{assumption}
\label{assump:3}
    Along the non-stationary iterations after identification, \[ \frac{\|\bm e_k^S\|} {\|\bm p_k^S\|} \to 0, \quad \frac{ \|(\bm B_k^S-\bm H_k^S)\bm p_k^S\| }{ \|\bm p_k^S\| } \to 0. \]
\end{assumption}

\begin{theorem}
\label{thm:5} 
Suppose Assumptions \ref{ass:main-1}--\ref{assump:3} hold. Then there exists $K_{\mathrm{full}}<\infty$ such that the full step is accepted $\lambda_k=\beta_k$ and $\bm x_{k+1} = \bm x_k+\bm p_k$ for $k\ge K_{\mathrm{full}}$. 
 Moreover, 
 \begin{equation}
     \lim_{k\to\infty} \frac{ \|\bm x_{k+1}-\bm x^\star\| }{ \|\bm x_k-\bm x^\star\| } = 0.
     \label{eq:thm-5-1}
 \end{equation}
\end{theorem}
\begin{proof}
    See Appendix \ref{appendix-J}.
\end{proof}

\section{Numerical Results}
We evaluate the efficiency and robustness of the proposed method (SPIN) in three sets of experiments. We first use LASSO problems under varied settings to assess robustness and scalability. We then study nonconvex sparse models with convex or nonconvex smooth losses and SCAD, MCP, or CEL0 penalties, comparing convergence speed and stationary-point quality. Finally, we evaluate the method on ridge-regularized hinge-loss SVMs using public datasets. We compare the proposed algorithm with IRPNM \cite{vomdahl2024inexact}, OESOM \cite{delosreyes2017secondorder}, and ANCG \cite{cheng2021active}, which represent the main second-order approaches reviewed in Section~\ref{subsec:I-A}. Comparisons with first-order methods have been reported in these studies and are not repeated here. We measure stationarity using the residual in \eqref{eq:proximal-gradient-residual} and terminate each algorithm when $\|\bm{G}_{\sigma}(\bm{x}_k)\|< {\texttt{tol}}$. We typically set \(\texttt{tol}<10^{-8}\) for high accuracy, adjusting it by test when needed. We set \(\sigma\) to the Lipschitz constant of \(\nabla q\). For each baseline, we use the recommended settings or code defaults. For SPIN, we set $\nu_k =\max\{10^{-5},\left\|\bm G_\sigma\left(\bm x_k\right)\right\|^{0.5}\}$, $\bar \vartheta = 0.25$, $\vartheta_k=\min \left\{\bar \vartheta, \left\|\bm G_\sigma\left(\bm x_k\right)\right\|^{0.5}\right\}$,  $\gamma = 2$, and $\eta = 10^{-4}$. $\bm B_k=\bm H_k+ \max \left\{0,0.05-\lambda_{\min }\left(\bm H_k\right)\right\} \bm I$. The box-constrained QP \eqref{eq:exact-Newton-direction} is solved by projected gradient initialized at zero, maintaining feasibility by projection onto $C_k$. All tested algorithms share the same initial point $\bm{x}_0$.

\begin{figure*}[htbp]
    \centering

    \subfloat[\label{fig:new_column_a}]{%
        \begin{minipage}[b]{0.250\textwidth}
            \centering
            \includegraphics[width=\linewidth]{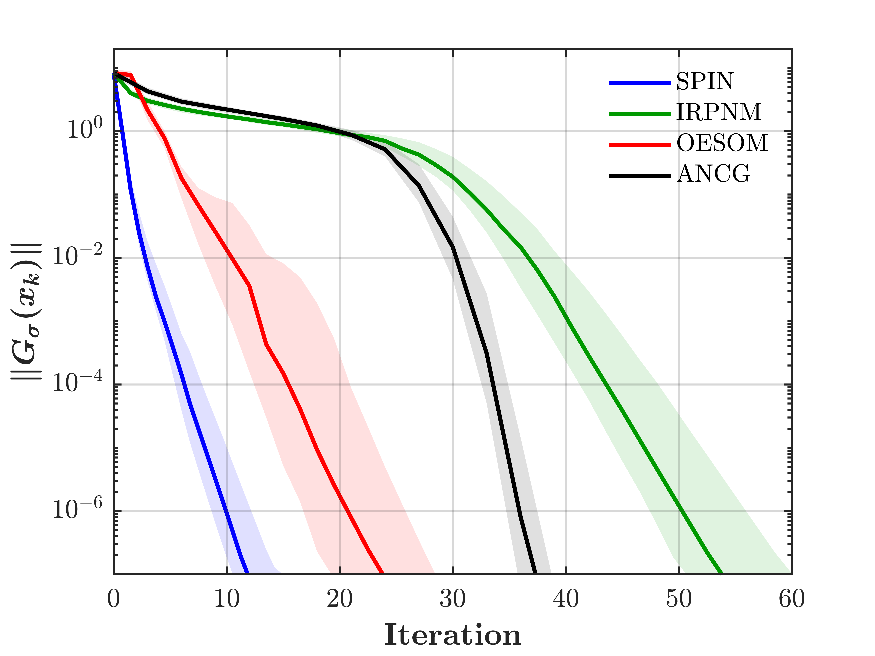}\\[-0.35mm]
            \includegraphics[width=\linewidth]{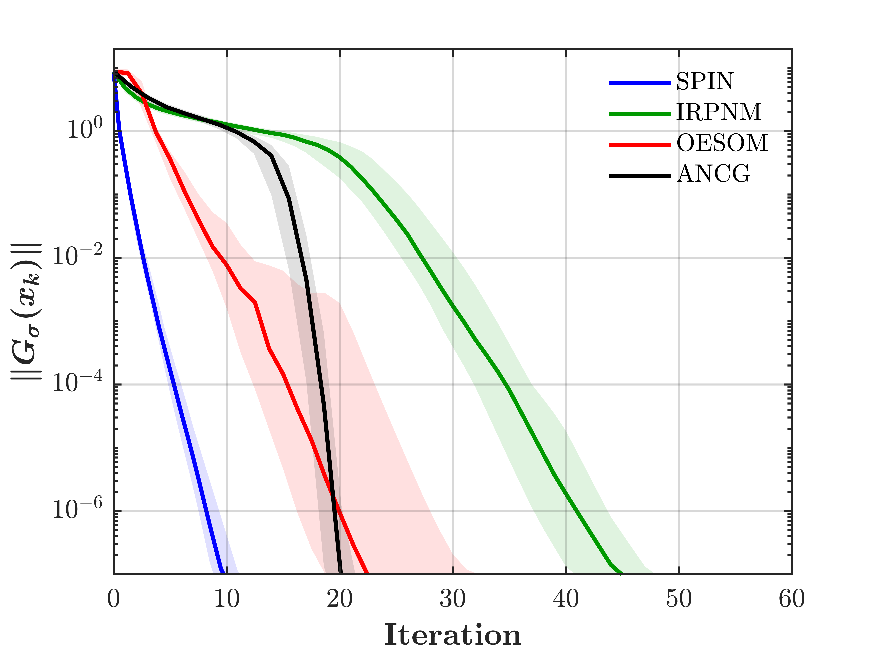}\\[-0.35mm]
            \includegraphics[width=\linewidth]{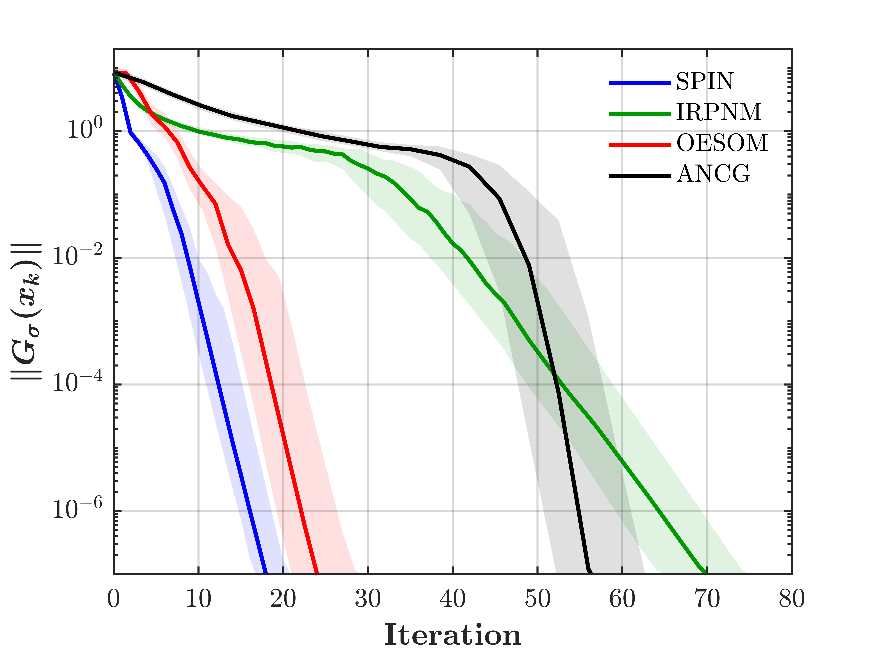}
        \end{minipage}%
    }%
    \hfill
    \subfloat[\label{fig:new_column_b}]{%
        \begin{minipage}[b]{0.250\textwidth}
            \centering
            \includegraphics[width=\linewidth]{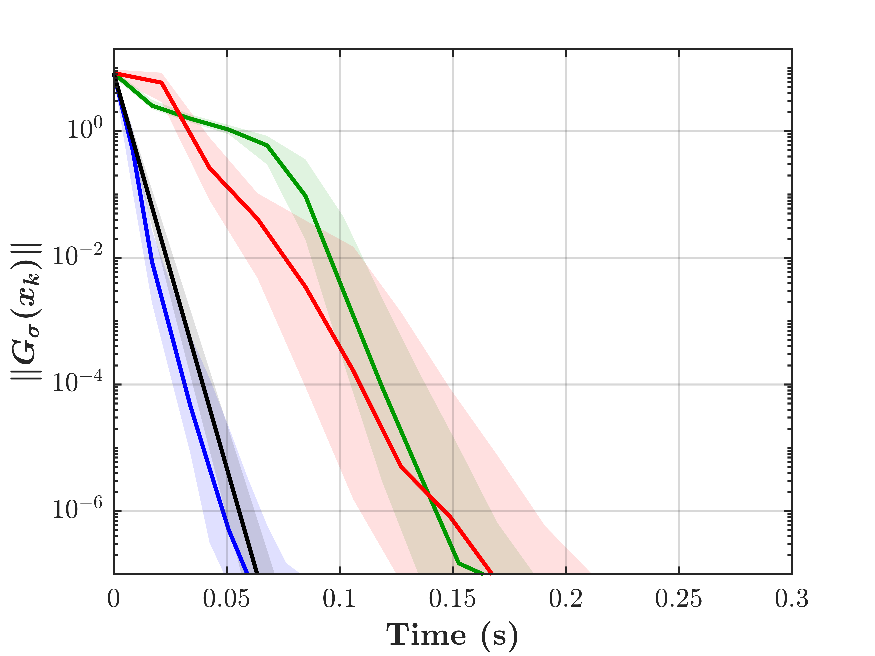}\\[-0.35mm]
            \includegraphics[width=\linewidth]{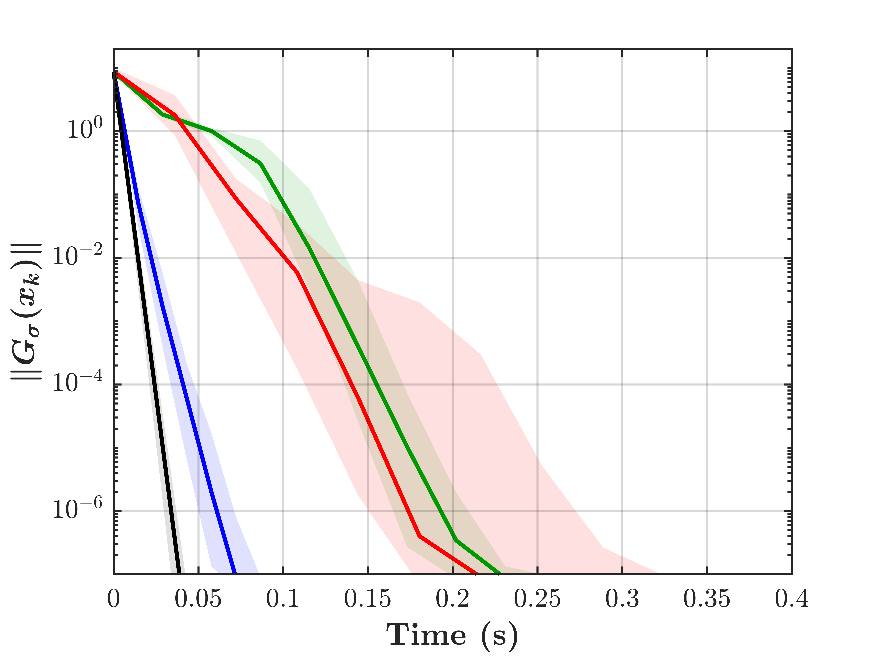}\\[-0.35mm]
            \includegraphics[width=\linewidth]{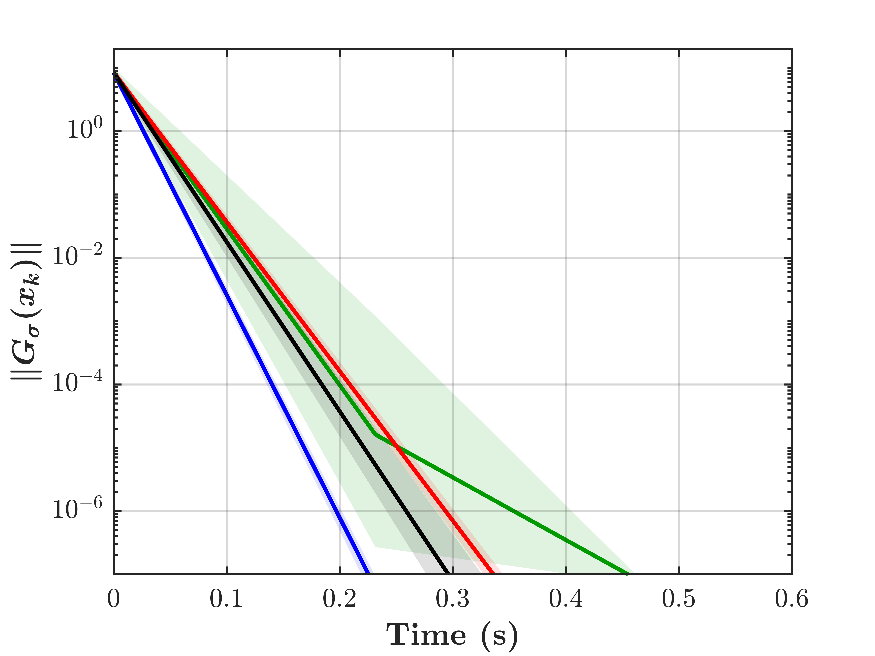}
        \end{minipage}%
    }%
    \hfill
    \subfloat[\label{fig:new_column_c}]{%
        \begin{minipage}[b]{0.250\textwidth}
            \centering
             \includegraphics[width=\linewidth]{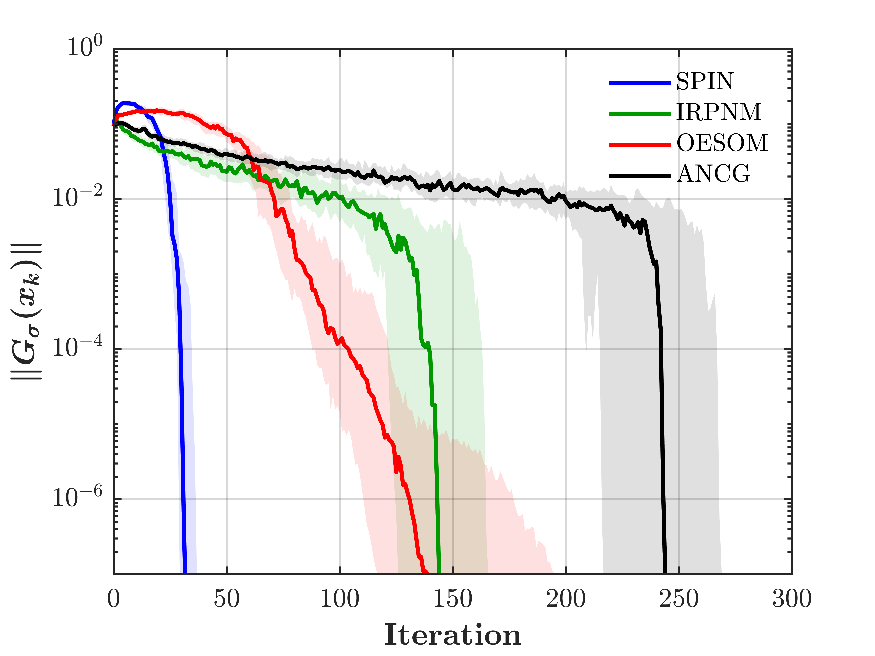}\\[-0.35mm]
            \includegraphics[width=\linewidth]{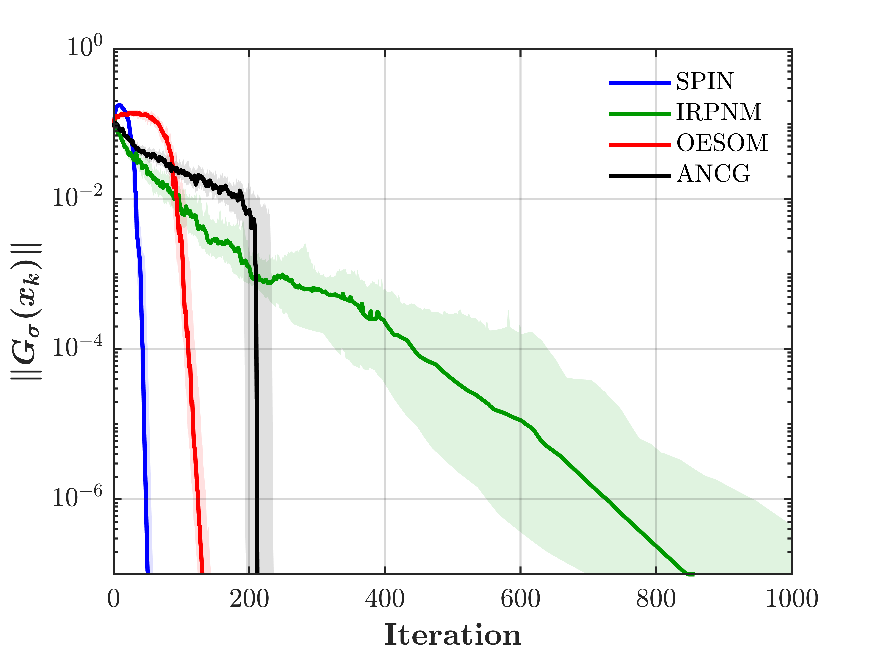}\\[-0.35mm]
            \includegraphics[width=\linewidth]{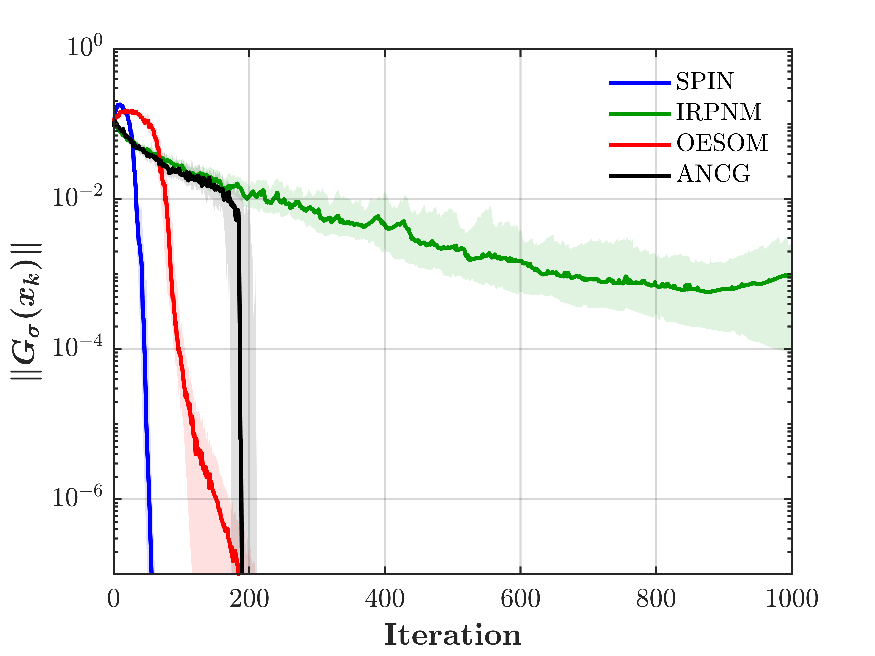}
        \end{minipage}%
    }%
    \hfill
    \subfloat[\label{fig:new_column_d}]{%
        \begin{minipage}[b]{0.250\textwidth}
            \centering
            \includegraphics[width=\linewidth]{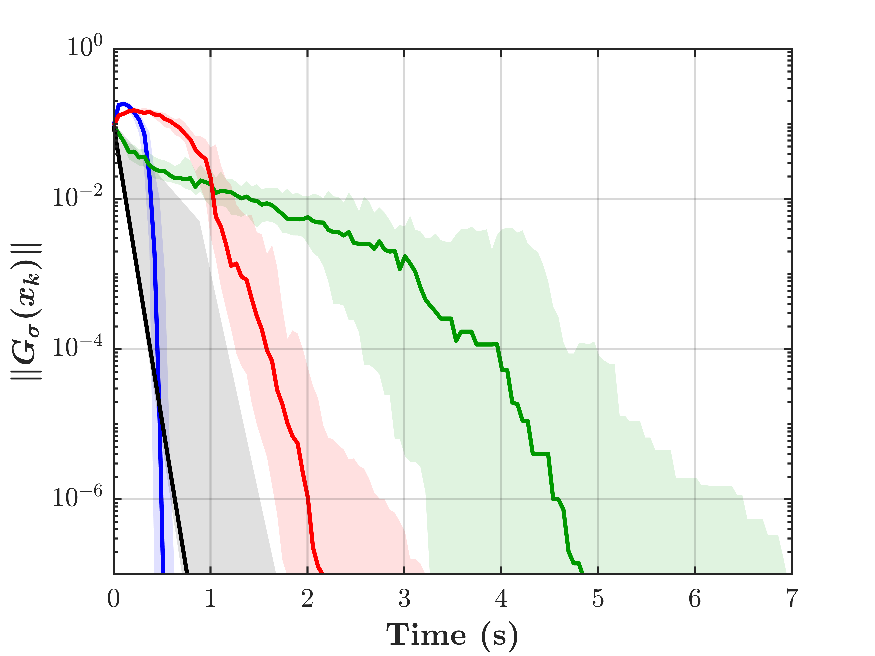}\\[-0.35mm]
            \includegraphics[width=\linewidth]{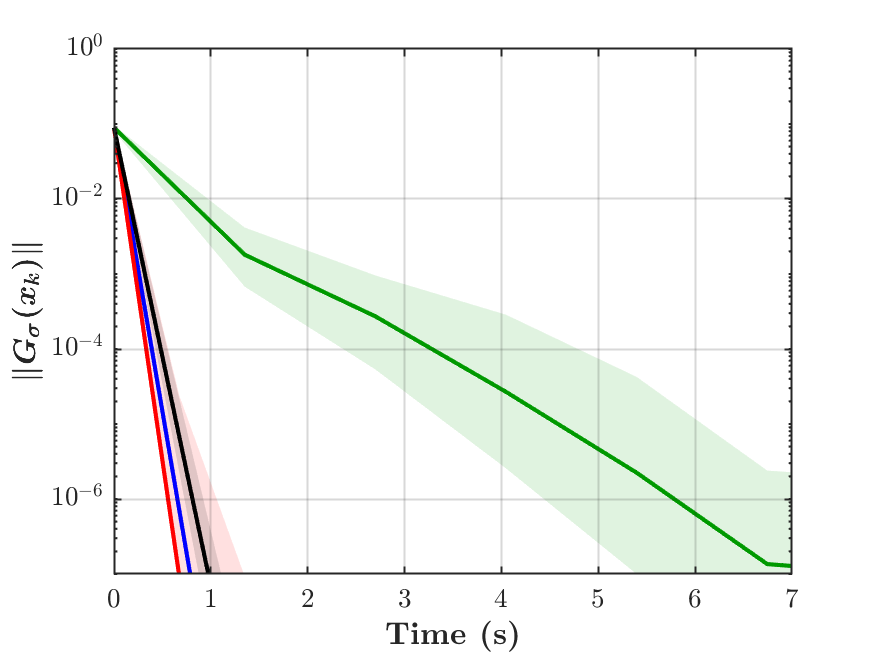}\\[-0.35mm]
            \includegraphics[width=\linewidth]{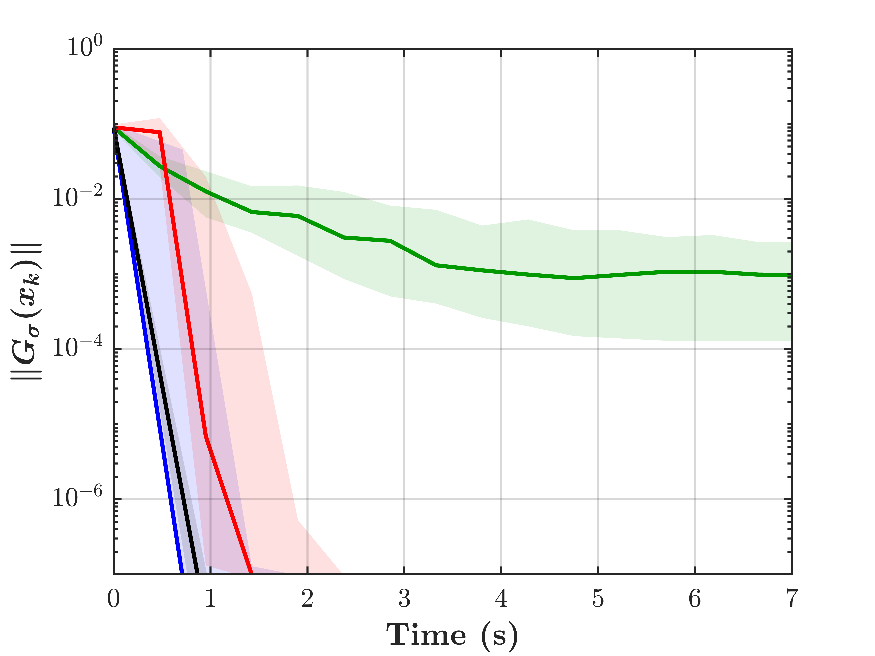}
        \end{minipage}%
    }

    \caption{Convergence of 6 nonconvex models. Rows show SCAD, MCP, and CEL0 from top to bottom. Columns (a) and (b) plot \(\|\bm{G}_\sigma(\bm{x}_k)\|\) against outer iteration count and elapsed time, respectively, for the convex least-squares loss. Columns (c) and (d) show the corresponding results for the nonconvex Cauchy loss. Results for each loss–penalty pair are aggregated over \(60\) independent realizations. Solid curves show the median across realizations at each iteration or aligned time point, and shaded bands span the corresponding 25th–75th percentiles.}
    \label{fig:2}
\end{figure*}

\subsection{ Convex sparse recovery}
We consider the convex sparse recovery
\begin{equation}
    \min_{\bm{x}\in\mathbb{R}^n}~\frac{1}{2}\|\bm{A}\bm{x}-\bm{b}\|^2_2 + \zeta\|\bm{x}\|_1,
    \label{eq:lasso}
\end{equation}
where $\bm{A}\in\mathbb{R}^{m\times n}$, and $\zeta>0$. In the default setting shown in Figure~\ref{fig:1}(a), \(\bm{A}\in\mathbb{R}^{400\times 800}\) has i.i.d. Gaussian entries with mean zero and variance \(0.01\). The true signal \(\bm{x}^\star\) has \(\kappa = 40\) nonzero entries with values sampled using Julia’s \(\texttt{rand()}\) function. The same random function is used to generate $\bm{x}_0$. We choose $\bm{b} = \bm A \bm{x}^\star + \bm \epsilon$ where components of $\bm \epsilon$ are i.i.d Gaussian noises with a
standard deviation $0.01$. The regularization parameter is $\zeta = 0.1\|\bm{A}^{\top}\bm b\|_\infty$.

Figure~\ref{fig:1} compares the convergence behavior of the proposed method and the baselines across five cases that vary the initialization, problem size, regularization parameter, and ground-truth sparsity. SPIN converges rapidly in both iteration count and runtime and remains robust across problem settings, whereas each baseline is sensitive to at least one problem factor. ANCG performs comparably to SPIN in panels (a), (c), and (e), but slows substantially under a distant initialization or small \(\zeta\), as shown in panels (b) and (d). In these cases, it requires many iterations to identify the correct active set. SPIN typically does so in only a few iterations.

\begin{table}[htbp]
\centering
\caption{Performance comparison of the solvers for SVM classification on real-world datasets.}
\label{tab:svm-results}

\scriptsize
\setlength{\tabcolsep}{1.8pt}
\renewcommand{\arraystretch}{0.76}

\begin{tabular*}{\columnwidth}
{@{\extracolsep{\fill}}l l c c@{}}
\toprule
\makecell[l]{Dataset\\$(n_{\rm tr},d)$}
& Algo.
& \makecell{Total\\$k/t(s)$}
& \makecell{Identification \\$k/t(s)$} \\
\midrule

\dataset{\texttt{australian}}{483}{14}
& SPIN
& \textbf{7}\,/\,\textbf{\num{0.092}}
& \textbf{6}\,/\,\textbf{\num{0.079}} \\
& IRPNM
& 9\,/\,\num{0.692}
& \textbf{6}\,/\,\num{0.361} \\
& OESOM
& 44\,/\,\num{0.294}
& 18\,/\,\num{0.125} \\
& ANCG
& 146\,/\,\num{0.193}
& 146\,/\,\num{0.193} \\
\midrule

\dataset{\texttt{diabetes}}{538}{8}
& SPIN
& \textbf{10}\,/\,\textbf{\num{0.241}}
& 9\,/\,{\num{0.215}} \\
& IRPNM
& \textbf{10}\,/\,\num{1.332}
& \textbf{6}\,/\,\num{0.586} \\
& OESOM
& 46\,/\,\num{0.424}
& 20\,/\,\textbf{\num{0.213}} \\
& ANCG
& 124\,/\,{\num{0.297}}
& 123\,/\,\num{0.294} \\
\midrule

\dataset{\texttt{fourclass}}{603}{2}
& SPIN
& \textbf{11}\,/\,\num{0.597}
& 10\,/\,\num{0.539} \\
& IRPNM
& 13\,/\,\num{2.785}
& \textbf{7}\,/\,\num{1.049} \\
& OESOM
& 56\,/\,\num{0.534}
& 15\,/\,\textbf{\num{0.139}} \\
& ANCG
& 84\,/\,\textbf{\num{0.179}}
& 84\,/\,\num{0.179} \\
\midrule

\dataset{\texttt{german}}{700}{24}
& SPIN
& \textbf{7}\,/\,\textbf{\num{0.320}}
& \textbf{6}\,/\,\textbf{\num{0.307}} \\
& IRPNM
& 9\,/\,\num{1.852}
& {7}\,/\,\num{1.195} \\
& OESOM
& 61\,/\,\num{0.824}
& 34\,/\,\num{0.476} \\
& ANCG
& 148\,/\,\num{0.648}
& 148\,/\,\num{0.648} \\
\midrule

\dataset{\texttt{ionosphere}}{246}{33}
& SPIN
& \textbf{6}\,/\,\num{0.020}
& 5\,/\,\textbf{\num{0.016}} \\
& IRPNM
& 9\,/\,\num{0.150}
& \textbf{4}\,/\,\num{0.046} \\
& OESOM
& 38\,/\,\num{0.082}
& 12\,/\,\num{0.025} \\
& ANCG
& 62\,/\,\textbf{\num{0.017}}
& 61\,/\,{\num{0.017}} \\
\midrule

\dataset{\texttt{splice}}{700}{60}
& SPIN
& \textbf{5}\,/\,\num{0.240}
& \textbf{5}\,/\,\num{0.240} \\
& IRPNM
& 10\,/\,\num{2.421}
& \textbf{5}\,/\,\num{0.849} \\
& OESOM
& 39\,/\,\num{0.500}
& 10\,/\,\num{0.127} \\
& ANCG
& 34\,/\,\textbf{\num{0.098}}
& 34\,/\,\textbf{\num{0.098}} \\
\midrule

\dataset{\texttt{leukemia}}{51}{7129}
& SPIN
& \textbf{3}\,/\,\textbf{\num{1.66e-04}}
& \textbf{3}\,/\,\textbf{\num{1.66e-04}} \\
& IRPNM
& 6\,/\,\num{8.24e-04}
& 4\,/\,\num{5.06e-04} \\
& OESOM
& 211\,/\,\num{0.0063}
& 30\,/\,\num{8.58e-04} \\
& ANCG
& 11\,/\,\num{2.02e-04}
& 10\,/\,\num{1.87e-04} \\
\midrule

\dataset{\texttt{a1a}}{1123}{110}
& SPIN
& \textbf{7}\,/\,\num{0.727}
& \textbf{6}\,/\,\num{0.635} \\
& IRPNM
& 9\,/\,\num{3.603}
& \textbf{6}\,/\,\num{1.833} \\
& OESOM
& 50\,/\,\num{1.389}
& 11\,/\,\textbf{\num{0.310}} \\
& ANCG
& 75\,/\,\textbf{\num{0.403}}
& 74\,/\,\num{0.394} \\
\midrule

\dataset{\texttt{svmguide3}}{899}{21}
& SPIN
& \textbf{8}\,/\,\textbf{\num{0.664}}
& \textbf{7}\,/\,\num{0.612} \\
& IRPNM
& 14\,/\,\num{6.998}
& {8}\,/\,\num{2.806} \\
& OESOM
& 41\,/\,\num{0.786}
& 17\,/\,\textbf{\num{0.328}} \\
& ANCG
& 77\,/\,{\num{0.742}}
& 77\,/\,\num{0.742} \\
\midrule

\dataset{\texttt{dna-ovr}}{1400}{180}
& SPIN
& \textbf{3}\,/\,\num{2.511}
& \textbf{3}\,/\,\num{2.511} \\
& IRPNM
& 14\,/\,\num{13.66}
& 6\,/\,\num{3.401} \\
& OESOM
& 56\,/\,\num{3.073}
& 9\,/\,\num{0.454} \\
& ANCG
& 130\,/\,\textbf{\num{0.967}}
& 129\,/\,\textbf{\num{0.832}} \\
\midrule

\dataset{\texttt{vehicle-ovr}}{592}{18}
& SPIN
& \textbf{8}\,/\,\textbf{\num{0.358}}
& 8\,/\,\textbf{\num{0.358}} \\
& IRPNM
& 10\,/\,\num{1.752}
& \textbf{6}\,/\,\num{0.890} \\
& OESOM
& 41\,/\,\num{0.489}
& 13\,/\,\num{0.172} \\
& ANCG
& 61\,/\,{\num{0.456}}
& 60\,/\,{\num{0.452}} \\
\midrule

\dataset{\texttt{vowel-ovr}}{693}{10}
& SPIN
& 13\,/\,\textbf{\num{0.772}}
& 12\,/\,\num{0.726} \\
& IRPNM
& \textbf{12}\,/\,\num{3.477}
& \textbf{8}\,/\,\num{1.953} \\
& OESOM
& 68\,/\,\num{1.286}
& 27\,/\,\textbf{\num{0.360}} \\
& ANCG
& 168\,/\,\num{1.136}
& 168\,/\,\num{1.136} \\
\midrule

\dataset{\texttt{svmguide1}}{4962}{4}
& SPIN
& \textbf{10}\,/\,\textbf{\num{22.8}}
& 9\,/\,{\num{20.9}} \\
& IRPNM
& 13\,/\,\num{154}
& \textbf{7}\,/\,\num{62.5} \\
& OESOM
& 50\,/\,\num{35.9}
& 14\,/\,\textbf{\num{9.7}} \\
& ANCG
& 124\,/\,{\num{27.0}}
& 123\,/\,\num{26.6} \\
\midrule

\dataset{\texttt{mushrooms}}{5687}{112}
& SPIN
& \textbf{6}\,/\,\textbf{\num{29.7}}
& \textbf{5}\,/\,{\num{25.4}} \\
& IRPNM
& 10\,/\,\num{106}
& \textbf{5}\,/\,\num{42.9} \\
& OESOM
& 42\,/\,\num{31.3}
& 10\,/\,\textbf{\num{7.04}} \\
& ANCG
& 138\,/\,{\num{32.8}}
& 137\,/\,{\num{32.6}} \\
\midrule

\dataset{\texttt{phishing}}{7739}{68}
& SPIN
& \textbf{8}\,/\,\num{31.6}
& 8\,/\,\num{31.6} \\
& IRPNM
& 12\,/\,\num{244}
& \textbf{6}\,/\,\num{98.2} \\
& OESOM
& 55\,/\,\num{71.3}
& 15\,/\,\textbf{\num{18.9}} \\
& ANCG
& 80\,/\,\textbf{\num{25.1}}
& 79\,/\,\num{24.9} \\
\bottomrule
\end{tabular*}
\end{table}
\subsection{Nonconvex sparse recovery}

We next consider six nonconvex models of the form
\begin{equation}
\min_{\bm{x}\in\mathbb{R}^{n}}
\widetilde q_{\ell}(\bm{x})
+\sum_{i=1}^{n}\phi_p(|x_i|),
\label{eq:nonconvex-lasso}
\end{equation}
where \(\ell\) indexes the least-squares (\(\mathrm{LS}\)) and Cauchy (\(\mathrm{Cau}\)) losses, \(p\) indexes the SCAD, MCP, and CEL0 penalties, and all six pairings are considered.  The least-squares loss \(\widetilde q_{\mathrm{LS}}\) takes the form given in \eqref{eq:lasso}, while the Cauchy loss is defined by
\begin{equation}
\widetilde q_{\mathrm{Cau}}(\bm{x})=
\frac{1}{m}\sum_{j=1}^{m}
\frac{\delta^2}{2}
\log\left(
1+\frac{(\bm{a}_j^\top\bm{x}-b_j)^2}{\delta^2}
\right),
\end{equation}
where \(\bm{a}_j^\top\) denotes the \(j\)th row of \(\bm{A}\). The Cauchy loss limits the influence of large residuals and is nonconvex because its scalar loss has negative curvature whenever $|\bm{a}_j^\top\bm{x}-b_j|>\delta$. Problem \eqref{eq:nonconvex-lasso} can be directly cast in the form \eqref{eq:problem-formulation} handled by all tested algorithms using objective decompositions described in Section~\ref{sec:I}. The common experimental setup is summarized as follows. For each realization, we generate \(\bm{A}\in\mathbb{R}^{200\times300}\) with independent standard Gaussian entries and normalize each column to have Euclidean norm \(\sqrt{m}\). The ground-truth signal \(\bm{x}^{\star}\) has \(\kappa=30\) randomly selected nonzero entries with values \(2z_i\), where $z_i$ are independently sampled from \(\mathcal{N}(0,1)\). The  observation is generated as $\bm{b}=
\bm{A}\bm{x}^{\star}+\bm{\epsilon}$, where \(\epsilon_j\) are independently sampled from \(\mathcal{N}(0,0.05^2)\).
For the Cauchy-loss models, we further contaminate \(15\%\) of the observations by adding perturbations of magnitude \(10\) with random signs. We set $\delta=0.5$ and $\zeta=
0.2\left\|
\nabla\widetilde q_{\ell}(\bm{0})
\right\|_{\infty}$.
We use \(\mu_i=\zeta\) and \(a_i=3.7\) for SCAD, and \(\mu_i=\zeta\) and \(\gamma_i=3\) for MCP. For CEL0, the column normalization and \(1/m\) loss scaling imply \(d_i=1\), so we set \(\mu=\zeta^2/2\), yielding \(d_i\sqrt{2\mu}=\zeta\). For each loss--penalty combination, we generate \(60\) independent realizations. Within each realization, all methods use the same data and the random initialization.

\begin{figure*}[htbp]
    \centering

    \includegraphics[width=0.3333\textwidth]{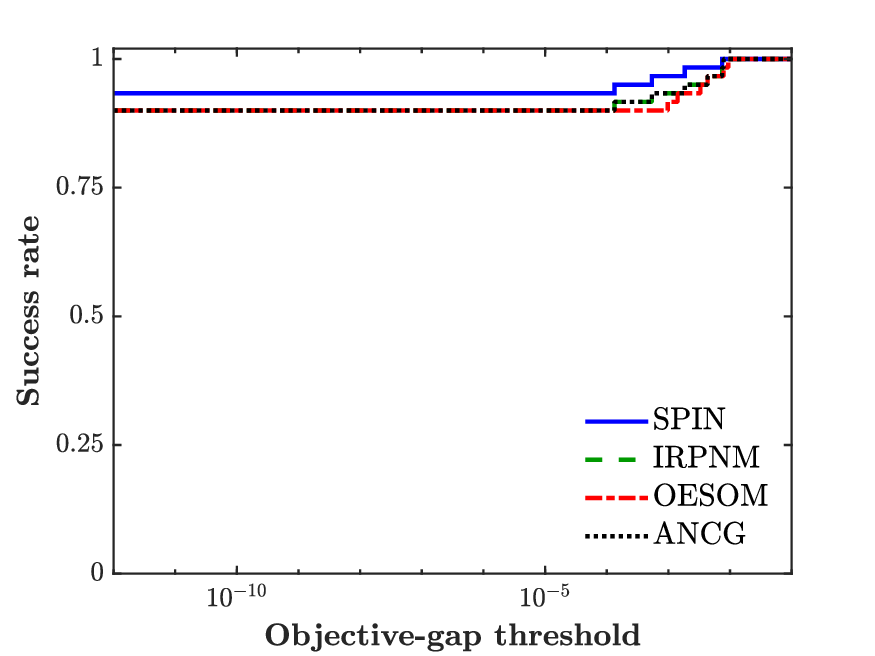}%
    \hfill
    \includegraphics[width=0.3333\textwidth]{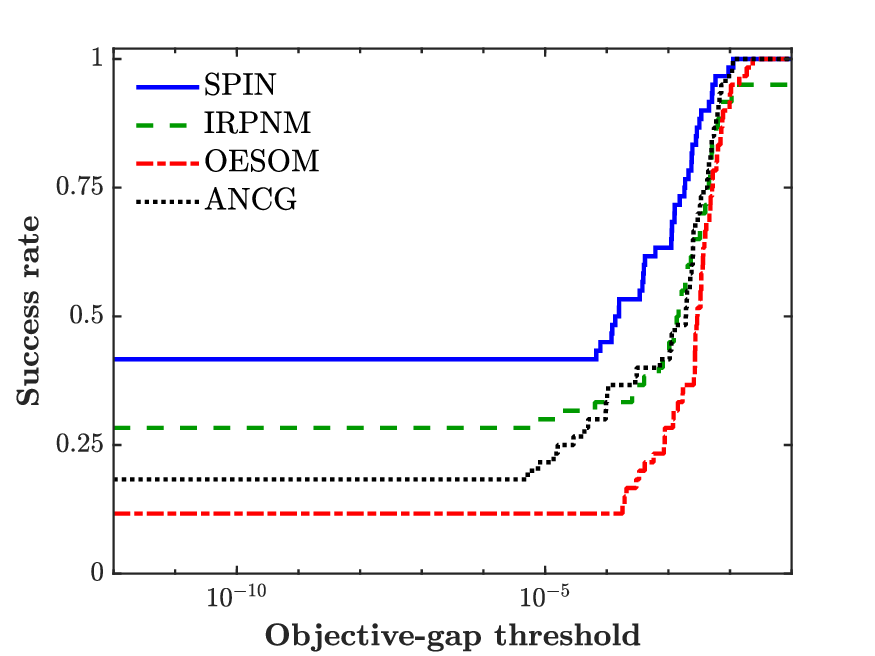}%
    \hfill
    \includegraphics[width=0.3333\textwidth]{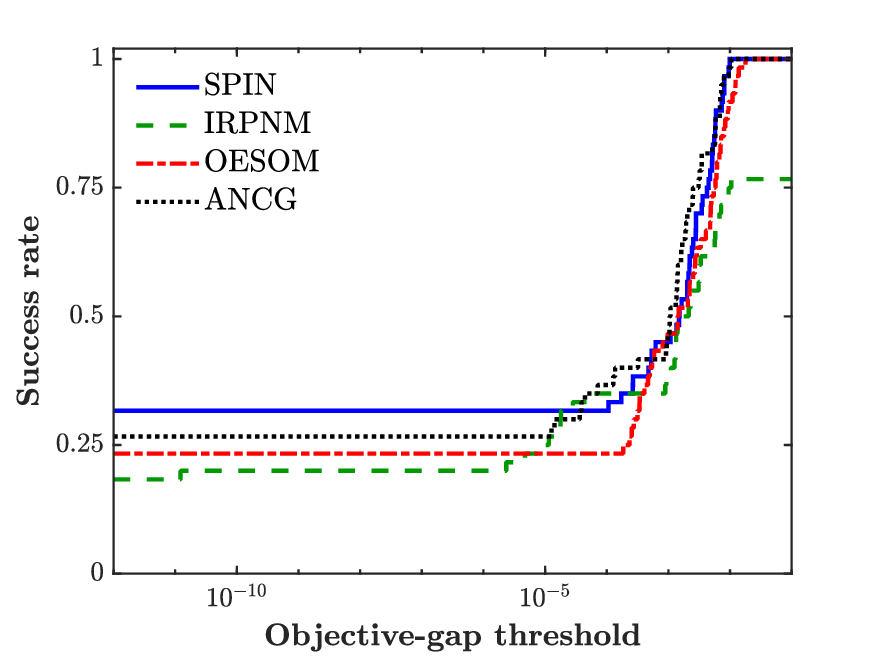}%

    \caption{Cumulative success rates for three nonconvex models. \textit{Left}: $\mathrm{LS}+\mathrm{CEL0}$. \textit{Middle}: $\mathrm{Cauchy}+\mathrm{SCAD}$. \textit{Right}: $\mathrm{Cauchy}+\mathrm{MCP}$.}
    \label{fig:3}
\end{figure*}

Figure~\ref{fig:2} examines convergence on six nonconvex models under a relatively benign regularization setting, using \(60\) independent realizations per model. SPIN converges consistently across all cases. Figure~\ref{fig:3} compares the quality of the local solutions returned by the algorithms through their cumulative success rates over \(60\) realizations. For an algorithm, realization \(r\) is successful at tolerance \(\epsilon\) if its final iterate \(\bm{x}\) satisfies \(\left\|\bm{G}_\sigma(\bm{x})\right\|\le 10^{-6},~\text{and }
\left(F(\bm{x})-F_r^{\mathrm{best}}\right)/
{\max\{1,|F_r^{\mathrm{best}}|\}}
\le \epsilon,\) where \(F_r^{\mathrm{best}}\) is the lowest final objective attained on realization \(r\) by any method satisfying the stationarity criterion.

\subsection{Cost-Sensitive Kernel SVMs}

We finally consider the cost-sensitive kernel SVM
\begin{equation}
\min_{f\in\mathcal H_K,\;b\in\mathbb R}
\frac{\lambda_f}{2}\|f\|_{\mathcal H_K}^{2}
+\frac{\lambda_b}{2}b^{2}
+\sum_{i=1}^{n} C_i
\bigl[\rho_i-y_i\bigl(f(\boldsymbol a_i)+b\bigr)\bigr]_{+},
\nonumber
\end{equation}
where $\mathcal H_K$ is the RKHS induced by an RBF kernel, $C_i$ are class-balanced weights, and $\rho_i$ are class-dependent margin targets. Following the change-of-variable strategy in Section~I, the problem is cast into the composite form~\eqref{eq:problem-formulation} for all tested methods. We use 15 public datasets from the LIBSVM Data repository \cite{libsvmdata}, retaining all samples and applying stratified \(70/30\) training--test splits. Multiclass datasets are converted into one-vs-rest binary classification tasks. Features are standardized using the training data, and the kernel bandwidth is selected by the median-distance rule. We set $\lambda_f=\lambda_b=10^{-2}$ and $\rho_i=1.25$ for the minority class and $\rho_i=1$ otherwise. Table \ref{tab:svm-results} reports the iteration counts and times required to satisfy \(\left\|\bm{G}_\sigma(\bm{x})\right\|\le 10^{-7}\) and identify the final active set.

\vspace{-0.5em}

\section{Conclusion}
\vspace{-0.5em}
We proposed an inexact proximal quasi-Newton method that combines active-set identification with Newton-type acceleration for nonconvex composite problems with separable convex polyhedral terms. It predicts the working set using first-order tests and Hessian coupling, then solves a constrained quadratic model inexactly while enforcing the structure constraints exactly. A proximal line search globalizes the step and tends to preserve the nonsmooth structure reached at the model endpoint. We proved finite termination of backtracking and stationarity of every accumulation point. Under strict complementarity and nonsingularity of the reduced Hessian, the full sequence converges and identifies the active manifold in finitely many iterations. A directional Dennis--Moré condition and vanishing inexactness further guarantee eventual full-step acceptance and local \(Q\)-superlinear convergence. Numerical experiments demonstrate the practical effectiveness of the method on the tested problems. 

\appendices
\vspace{-0.5em}
\section{Proof of Lemma~\ref{lem:1}}
\label{appendix-A}
\begin{proof}
    The normal cone of a Cartesian box is the product of its scalar normal cones. For a nondegenerate interval, the scalar normal is $\{0\}$ in the interior, $(-\infty, 0]$ at the lower endpoint, and $[0,+\infty)$ at the upper endpoint. Therefore the Euclidean projection of $-q_i$ onto the scalar normal leaves the residual $q_i, \min \left\{q_i, 0\right\}$, or $\max \left\{q_i, 0\right\}$, respectively. A degenerate interval has normal cone $\mathbb{R}$ and hence zero residual. Squaring and summing the coordinate-wise distances proves \eqref{eq:inexact-condition-3}. The bounded $\bm{B}^W_k$ gives $\varepsilon_k\left(\bm{p}_k\right) \leq m^{-1 / 2}\left\|\bm{E}_k\right\| \leq \vartheta_k \sqrt{m}\left\|\bm{p}_{k, W_k}\right\| \leq \vartheta_k\left\|\bm{p}_{k, W_k}\right\|_{\bm{B}_k^W},$ which is \eqref{eq:inexact-condition}.
\end{proof}

\section{Proof of Theorem~\ref{the:1}}
\label{appendix-B}
\begin{proof}
Since \(h\) is separable, \eqref{eq:main-step} decomposes into independent scalar subproblems. Fix \(i\) and, for brevity, write $x=x_{k,i},~
p=p_{k,i},~
y=x+p,~
z=z_{k,i}(\lambda),~
v=v_{k,i}$.
By the definition of $\bm c_k$, the scalar optimality condition is \begin{equation}
    0
\in
\lambda(z-x)-\beta_kp+\partial h_i(z)-v.
\end{equation}
If \(p=0\), then \(v\in\partial h_i(x)\), so \(z=x\) satisfies this condition and is therefore the unique scalar proximal minimizer. If $p\ne0$, monotonicity of $\partial h_i$ implies
\[
0\le \operatorname{sign}(p)(z-x-\frac{\beta_k}{\lambda}p)\le \operatorname{sign}(p)(y-x-\frac{\beta_k}{\lambda}p).
\]
Indeed, if $\operatorname{sign}(p)(z-x-\frac{\beta_k}{\lambda}p)<0$, then
$\operatorname{sign}(p)(\zeta-v)\le0$ for every $\zeta\in\partial h_i(z)$. If
$\operatorname{sign}(p)(z-y)>0$, then $\operatorname{sign}(p)(\zeta-v)\ge0$. Either case contradicts the
preceding inclusion. Thus, $z$ lies between $x+\frac{\beta_k}{\lambda}p$ and $x+p$.

If $i\notin D_k$, this segment lies in the affine interval and
$v=\alpha_{k,i}$ which gives $z = x+\frac{\beta_k}{\lambda}p$. Now let $i\in D_k$. There exists
$\kappa_{k,i}\ge0$ at $y$ such that
\[
\partial h_i(y)
=
\alpha_{k,i}+\operatorname{sign}(p)[0,\kappa_{k,i}],\quad
v
=
\alpha_{k,i}+\operatorname{sign}(p)\omega_{k,i}.
\]
The endpoint $z=y$ is optimal exactly when $v-\operatorname{sign}(p)(\lambda-\beta_k)|p|\in\partial h_i(y)$,
which is equivalent to  $\lambda
    \le
    \beta_k+{\omega_{k,i}}/{\abs{p_{k,i}}}$. Otherwise,
$z$ lies strictly before $y$, where
$\partial h_i(z)=\{\alpha_{k,i}\}$, and hence
\[
z-x
=
\frac{\beta_k+\omega_{k,i}/|p|}{\lambda}\,p.
\]
This proves \eqref{eq:the:1-1} and all stated geometric consequences. Assume \eqref{eq:inexact-condition} holds. Exact structure feasibility gives $\Delta_k(\bm p_k)
=
\bm a_k^\top\bm p_{k,W_k}$. 
Taking the inner product of
\[
\bm e_k
=
\bm a_k
+\bm B_k^W\bm p_{k,W_k}
+\bm n_k
\]
with $\bm p_{k,W_k}$, and using
$\bm n_k^\top\bm p_{k,W_k}\ge0$, gives
\begin{equation}
\begin{aligned}
\Delta_k(\bm p_k)
&=
\langle\bm e_k,\bm p_{k,W_k}\rangle
-\|\bm p_{k,W_k}\|_{\bm B_k^W}^2
-\langle\bm n_k,\bm p_{k,W_k}\rangle\\
&\le
-(1-\vartheta_k)
\|\bm p_{k,W_k}\|_{\bm B_k^W}^2=
-(1-\vartheta_k)\beta_k\|\bm p_k\|^2.
\end{aligned}
\label{eq:Delta-beta}
\end{equation}

We next build the bound on $\bm{d}_k(\lambda)$ to show \eqref{eq:the:1-2}. For $i\in D_k$, define
\[
\rho_{k,i}
:=
\min\!\left\{
\left(1-\frac{\beta_k}{\lambda}\right)|p|,
\frac{\omega_{k,i}}{\lambda}
\right\},
\]
and set $\rho_{k,i}=0$ for $i\notin D_k$. The coordinate formula gives
\[
\bm d_k(\lambda)
=
\frac{\beta_k}{\lambda}\bm p_k
+
\sum_{i\in D_k} \operatorname{sign}(p)\rho_{k,i}\bm e_i.
\]
All coordinates remain in the selected interval so
\begin{equation}
    \Delta_k\bigl(\bm d_k(\lambda)\bigr)
=
\frac{\beta_k}{\lambda}\Delta_k(\bm p_k)
+
\sum_{i\in D_k}
\operatorname{sign}(p)(g_{k,i}+\alpha_{k,i})\rho_{k,i}. \nonumber
\end{equation}
For $i\in D_k$ with $\rho_{k,i}>0$, the definition of $v_{k,i}$ \eqref{eq:vbar} gives
$\sign(p)(g_{k,i}+\beta_k p+ v_{k,i})\le0$.  Hence
\begin{align}
    -\sign(p)\paren{g_{k,i}+\alpha_{k,i}}
    &=\beta_k\abs{p}
      +\omega_{k,i}-
      \sign(p)(g_{k,i}+\beta_k p+ v_{k,i})\nonumber\\
    &\ge\beta_k\abs{p}+\lambda \rho_{k,i}.\label{eq:extra-slope-bound}
\end{align}
For $a,b\ge0$ and $0\le\vartheta<1/2$,
\begin{align}
    &(1-\vartheta)a^2+ab+b^2\label{eq:algebra-identity}\\
    &\quad=\sigma(\vartheta)(a+b)^2
    +\frac{\paren{(1-2\vartheta)a-b}^2}{4(1-\vartheta)}
    \ge\sigma(\vartheta)(a+b)^2.\nonumber
\end{align}
Apply \eqref{eq:algebra-identity} coordinate-wise with
$a=(\beta_k/\lambda)\abs{p}$ and $b=\rho_{k,i}$, using \eqref{eq:Delta-beta}--\eqref{eq:extra-slope-bound}. Summation gives
$\Delta_k(\bm{d}_k(\lambda))\le-\sigma(\vartheta_k)\lambda\norm{\bm{d}_k(\lambda)}^2$.  Since $\vartheta_k\le\overline\vartheta<\frac{1}{2}$ and $\sigma(\vartheta)$ is decreasing on $[0,1/2)$, \eqref{eq:the:1-2} follows.
\end{proof}

\section{Proof of Lemma \ref{lem:inexact}}
\label{appendix-C}
\begin{proof}
    Since $\bm{B}^W_k \succeq m\bm{I}$, the reduced model is strongly convex, hence $\bm{p}^\star_{k,W_k}$ is unique. Choose $\bm n_k^{\star}$ satisfying $\bm a_k+\bm B_k^W \bm p_k^{\star}+\bm n_k^{\star}=0$. The definition of $\bm{e}_k$ and using monotonicity of the normal cone gives
    \begin{equation}
\begin{aligned}
&\left\langle \bm{e}_k, \bm{p}_{k,W_k}-\bm{p}^\star_{k,W_k}\right\rangle
=\left\|\bm{p}_{k,W_k}-\bm{p}^\star_{k,W_k}\right\|_{\bm{B}^W_k}^2\\
&+\left\langle \bm{n}_k-\bm{n}^{\star}_k, \bm{p}_{k,W_k}-\bm{p}_{k,W_k}^{\star}\right\rangle 
\geq\left\|\bm{p}_{k,W_k}-\bm{p}^\star_{k,W_k}\right\|_{\bm{B}^W_k}^2.
\end{aligned}
\nonumber
\end{equation}
Cauchy-Schwarz in $\bm{B}^W_k$-metric therefore yields \eqref{eq:lem-2-1}. Next, define $\phi(\bm{w}):=\bm{a}_k^{\top}\bm{w}+\frac{1}{2}\bm{w}^{\top}\bm{B}_k^W\bm{w}+\iota_{C_{k,W_k}}(\bm{w})$. The function is one-strongly convex in the $\bm{B}_k^W$ norm and $\bm{e}_k \in \partial \phi(\bm{p}_{k,W_k})$. Hence
\begin{equation}
\begin{aligned}
    &m_k\left(\bm{p}_k\right)-m_k\left(\bm{p}_k^{\star}\right) \leq \left\langle \bm{e}_k, \bm{p}_{k,W_k}-\bm{p}^\star_{k,W_k}\right\rangle \\
    &-\frac{1}{2} \left\|\bm{p}_{k,W_k}-\bm{p}^\star_{k,W_k}\right\|_{\bm{B}^W_k}^2
\leq \frac{1}{2}\left\|\bm{e}_k\right\|_{\left(\bm{B}_k^W\right)^{-1}}^2,
\end{aligned}
\end{equation}
where the last inequality is Young’s inequality. 

We then show the stationarity equivalence. Suppose first that $\bm{r}\left(\bm{x}_k\right)=\bm{0}$. Then $-\bm{g}_k \in \partial h\left(\bm{x}_k\right)$. By convexity of $h$, $h\left(\bm{x}_k+\bm{p}_k\right)-h\left(\bm{x}_k\right) \geq-\left\langle\bm{g}_k, \bm{p}_k\right\rangle$, and hence $\Delta_k\left(\bm{p}_k\right) \geq 0$. Combining this inequality with the preceding descent estimate gives $0 \leq \Delta_k\left(\bm{p}_k\right) \leq-\left(1-\vartheta_k\right)\left\|\bm{p}_{k, W_k}\right\|_{\bm{B}_k^W}^2 $ Therefore $\bm{p}_{k, W_k}=\bm{0}$. Since feasibility also imposes $\bm{p}_{k, N_k}=\bm{0}$, it follows that $\bm{p}_k=\bm{0}$. Conversely, suppose that $\bm{p}_k=\bm{0}$. \eqref{eq:inexact-condition} yields
$\varepsilon_k(\bm{0})=0$, so the origin satisfies the exact KKT conditions of the reduced box model. For every $i \in$ $S_k$, the origin is an interior point of the displacement interval, and therefore $g_{k, i}+\alpha_{k, i}=0$.  Now suppose, for contradiction, that $\delta_{k, i}<0$ for some
$i \in Z_k$. Then $i \in V_k \subseteq J_k$, so this coordinate is included in the working set. At a selected kink coordinate, the scalar KKT condition at the endpoint $p_i=0$ is $\xi_{k, i}\left(g_{k, i}+\alpha_{k, i}\right)=\delta_{k, i} \geq 0$ contradicting $\delta_{k, i}<0$. Thus
$\delta_{k, i} \geq 0$ for every $i \in Z_k$. Based on the definition of $\delta_{k, i}$, both directional derivatives are nonnegative. Together with the equalities on $S_k$,
Proposition \ref{prop:1} gives $\bm{r}\left(\bm{x}_k\right)=\bm{0}$.
\end{proof}

\section{Proof of Proposition \ref{prop:2}}
\label{appendix-D}
\begin{proof}
    Fix $j \in E_k$. Because $j$ is released only in the selected direction $\xi_{k, j}$, its feasible displacement is $p_j=\xi_{k, j} t$ with $ t \geq 0$. Hold all the other components fixed at their optimal values $p_{k, i}^{\star}$. As a function of $t$, the terms of the quadratic model involving coordinate $j$ are
    \begin{equation}
\phi_j(t)=\left(\delta_{k, j}+\xi_{k, j} \sum_{i \in W_k \backslash\{j\}}\left(\bm{B}_k\right)_{j i} p_{k, i}^{\star}\right) t+\frac{1}{2}\left(\bm{B}_k\right)_{j j} t^2 . \nonumber
\end{equation}
The function $\phi_j$ is strictly convex, so its minimizer satisfies $t>0$ exactly when its derivative at the origin is negative
\begin{equation}
\phi_j^{\prime}(0)=\delta_{k, j}+\xi_{k, j} \sum_{i \in W_k \backslash\{j\}}\left(\bm{B}_k\right)_{j i} p_{k, i}^{\star}<0 .
\end{equation}
Because $p_{k, j}^{\star}=\xi_{k, j} t$, this condition is equivalent to \eqref{eq:prop-2-1}.
\end{proof}

\section{Proof of Proposition \ref{prop:3}}
\label{appendix-E}
\begin{proof}
    For any $\bm{q} \in C_k$, the quadratic expansion of $m_k$ around $\bm{p}_k^\star$ gives
    \begin{equation}
\begin{aligned}
m_k(\bm{q})-m_k\left(\bm{p}_k^{\star}\right)= & \left\langle\nabla m_k\left(\bm{p}_k^{\star}\right)_{W_k}, \bm{q}_{W_k}-\bm{p}_{k, W_k}^{\star}\right\rangle \\
& +\frac{1}{2}\left\|\bm{q}_{W_k}-\bm{p}_{k, W_k}^{\star}\right\|_{\bm{B}_k^W}^2 .
\end{aligned}
\end{equation}
Because $\bm{p}_k^{\star}$ minimizes \eqref{eq:exact-Newton-direction}, we have $\left\langle\nabla m_k\left(\bm{p}_k^{\star}\right)_{W_k}, \bm{q}_{W_k}-\bm{p}_{k, W_k}^{\star}\right\rangle \geq 0$. Taking $\bm{q}=\widehat{\bm{p}}_k$ yields $m_k\left(\widehat{\bm{p}}_k\right)-m_k\left(\bm{p}_k^{\star}\right) \geq \frac{1}{2}\left\|\widehat{\bm{p}}_{k, W_k}-\bm{p}_{k, W_k}^{\star}\right\|_{\bm{B}_k^W}^2$. Suppose $\bm{p}_{k, E_k}^{\star} \neq \bm{0}$. Then $\bm{p}_k^{\star}$ is not feasible for the empty $E_k$ model, so $\widehat{\bm{p}}_k \neq \bm{p}_k^{\star}$.  Strong convexity therefore makes
\begin{equation}
m_k\left(\bm{p}_k^{\star}\right)<m_k\left(\widehat{\bm{p}}_k\right) .
\end{equation}
Finally, using \eqref{eq:lem-2-2} and $ m_k\left(\widehat{\bm{p}}_k\right)-m_k\left(\bm{p}_k\right)=  {\left[m_k\left(\widehat{\bm{p}}_k\right)-m_k\left(\bm{p}_k^{\star}\right)\right] }  -\left[m_k\left(\bm{p}_k\right)-m_k\left(\bm{p}_k^{\star}\right)\right]$ proves \eqref{eq:prop-3-1}.
\end{proof}

\section{Proof of Theorem \ref{the:2}}
\label{appendix-F}
\begin{proof}
For any trial value $\lambda \geq \beta_k$, the Lipschitz continuity of $\nabla q$ and the path-descent estimate \eqref{eq:the:1-2} give
\begin{equation}
\begin{aligned}
F\left(\bm{z}_k(\lambda)\right)-F\left(\bm{x}_k\right) & \leq \Delta_k\left(\bm{d}_k(\lambda)\right)+\frac{L}{2}\left\|\bm{d}_k(\lambda)\right\|^2 \\
& \leq-\left(\bar{\sigma} \lambda-\frac{L}{2}\right)\left\|\bm{d}_k(\lambda)\right\|^2 .
\end{aligned}
\nonumber
\end{equation}
If $\lambda \geq \frac{L}{2 \bar{\sigma}-\eta},$ then $\bar{\sigma} \lambda-\frac{L}{2} \geq \frac{\eta \lambda}{2}$. Hence, \eqref{eq:armijo} holds. Since the trial values are $\lambda_{k, j}=\beta_k \gamma^j,~ \gamma>1$, they eventually exceed the required threshold. Thus, backtracking terminates finitely. If $\beta_k$ already exceeds the threshold, the first trial is accepted and $\lambda_k=\beta_k \leq M .$  Otherwise, the first trial exceeding the threshold is smaller than $\frac{\gamma L}{2 \bar{\sigma}-\eta} .$  Therefore, $\lambda_k \leq \max \left\{M, \frac{\gamma L}{2 \bar{\sigma}-\eta}\right\}=\bar{\lambda}$. Finally, since $\beta_k \geq m$, $\frac{\beta_k}{\lambda_k} \geq \frac{m}{\bar{\lambda}}>0 .$ This completes the proof.
\end{proof}

\section{Proof of Theorem \ref{thm:3}}
\label{appendix-G}
\begin{proof}
The acceptance condition \eqref{eq:armijo}, The path bound \eqref{eq:the:1-1}, and the
uniform bound $\lambda_k\le\bar\lambda$ from Theorem~\eqref{the:2} give
\[
\begin{aligned}
&F(\bm x_k)-F(\bm x_{k+1})
\ge
\frac{\eta\lambda_k}{2}\|\bm d_k\|^2
\ge
\frac{\eta\lambda_k}{2}
\left(\frac{\beta_k}{\lambda_k}\right)^2
\|\bm p_k\|^2
\\
&=
\frac{\eta\beta_k^2}{2\lambda_k}\|\bm p_k\|^2\ge
\frac{\eta m^2}{2\bar\lambda}\|\bm p_k\|^2.
\end{aligned}
\]
This proves \eqref{eq:global-descent}. Since $F(\bm x_k)$ is nonincreasing and bounded below, $\sum_{k=0}^{\infty}\|\bm p_k\|^2<\infty$
and hence $\bm p_k\to\bm0$. \eqref{eq:the:1-2} and \eqref{eq:inexact-condition} imply $\|\bm d_k\|\le\|\bm p_k\|\to0$ and $\|\bm e_k\|
\le M\vartheta_k\|\bm p_k\|
\to0$. It remains to prove stationarity. Let $\bm x_{k_j}\to\bar{\bm x}$
be any convergent subsequence, and define $\bm y_k:=\bm x_k+\bm p_k$. Then $\bm y_{k_j}\to\bar{\bm x}$, and $\bm x_{k_j+1}\to\bar{\bm x}.$

Fix a coordinate $i$. If $\bar x_i\notin\mathcal B$, then, for all
sufficiently large $j$, both $x_{k_j,i}$ and $y_{k_j,i}$ lie in the
interior of the same affine interval of $h_i$. The corresponding normal
component is therefore zero, and the model residual relation gives
\[
e_{k_j,i}
=
g_i(\bm x_{k_j})
+
h_i'(\bar x_i)
+
(\bm B_{k_j}\bm p_{k_j})_i.
\]
Since $\|\bm B_{k_j}^W\|\le M$, passing to the limit yields $g_i(\bar{\bm x})+h_i'(\bar x_i)=0$.

Now suppose that $\bar x_i=b\in\mathcal B$. We show that $\xi\bigl(g_i(\bar{\bm x})+h'_{i,\xi}(b)\bigr)\ge0,~
\xi\in\{-1,+1\}.$
Assume, to the contrary, that for some $\xi\in\{-1,+1\}$,
\begin{equation}
\xi\bigl(g_i(\bar{\bm x})+h'_{i,\xi}(b)\bigr)<0.
\label{eq:kink-violation}
\end{equation}
After passing to a further subsequence, either $\xi(x_{k_j,i}-b)\ge0$
for every $j$, or $\xi(x_{k_j,i}-b)<0$
for every $j$. Consider first the case. 
For all sufficiently large $j$, the selected interval has slope
$h'_{i,\xi}(b)$. Moreover, since $\bm p_{k_j}\to\bm0$, the model
point cannot reach the far endpoint of this interval. Hence the
corresponding normal component satisfies $\xi n_{k_j,i}\le0.$
Multiplying the coordinate residual relation by $\xi$ gives
\begin{equation}
    \xi e_{k_j,i}
\le
\xi\bigl(g_i(\bm x_{k_j})+h'_{i,\xi}(b)\bigr)
+
M\|\bm p_{k_j}\|.\label{eq:thm-3-pro-1}
\end{equation}
The right-hand side converges to the negative quantity in
\eqref{eq:kink-violation}, contradicting $\bm e_{k_j}\to\bm0$. For the second case, if $y_{k_j,i}\neq b$ along an infinite subsequence, then, for all
sufficiently large $j$, the model point lies in the interior of the
opposite affine interval. Its normal component is zero, and its slope is
$h'_{i,-\xi}(b)$. Convexity of $h_i$ gives $\xi h'_{i,-\xi}(b)
\le
\xi h'_{i,\xi}(b)$.
Therefore, \eqref{eq:thm-3-pro-1} is obtained
which again contradicts \eqref{eq:kink-violation}. Consequently, $y_{k_j,i}=b$
for all sufficiently large $j$, and $\operatorname{sign}(p_{k_j,i})=\xi.$

Since $\beta_{k_j}\le M$ and $p_{k_j,i}\to0$, condition
\eqref{eq:kink-violation} implies
\[
\xi\left(
-g_i(\bm x_{k_j})
-\beta_{k_j}p_{k_j,i}
-h'_{i,\xi}(b)
\right)>0
\]
for all sufficiently large $j$. Hence the projection in~\eqref{eq:vbar} selects
the endpoint subgradient $v_{k_j,i}=h'_{i,\xi}(b).$
It follows that
\[
\omega_{k_j,i}
=
\xi\left(
h'_{i,\xi}(b)-h'_{i,-\xi}(b)
\right)
=
h'_{i,+}(b)-h'_{i,-}(b)
>0.
\]
Because $\lambda_{k_j}\le\bar\lambda$ and $p_{k_j,i}\to0$, we eventually
have
\[
\lambda_{k_j}
\le
\beta_{k_j}
+
\frac{\omega_{k_j,i}}{|p_{k_j,i}|}.
\]
Theorem~\ref{the:1} therefore gives $x_{k_j+1,i}=b.$
Applying the first case to the shifted subsequence
$\{\bm x_{k_j+1}\}$ yields the same contradiction. Thus both one-sided directional derivatives are nonnegative at every
kink coordinate. Together with the equality obtained at smooth
coordinates, Proposition~\ref{prop:1} gives $0\in\partial F(\bar{\bm x}).$
Hence every accumulation point is stationary.

Finally, suppose that \eqref{eq:distance-stationary} fails. Then there
exist $\varepsilon>0$ and a subsequence satisfying $\operatorname{dist}(\bm x_{k_j},\mathcal X)\ge\varepsilon.$
All iterates lie in the compact initial sublevel set, so a further
subsequence converges to some $\bar{\bm x}$. The preceding argument
gives $\bar{\bm x}\in\mathcal X$, and therefore $\operatorname{dist}(\bm x_{k_j},\mathcal X)
\le
\|\bm x_{k_j}-\bar{\bm x}\|
\to0,$
a contradiction. This proves \eqref{eq:distance-stationary}.
\end{proof}

\section{Proof of Proposition \ref{prop:4}}
\label{appendix-H}
\begin{proof}
We first prove that $\bm x^\star$ is isolated. For every
$i\in Z^\star$, strict complementarity gives $g_i(\bm x^\star)+h'_{i,+}(x_i^\star)>0,$ and $g_i(\bm x^\star)+h'_{i,-}(x_i^\star)<0.$ After shrinking a neighborhood of $\bm x^\star$, these inequalities
remain valid, and $x_i^\star$ is the only kink of $h_i$ in that
neighborhood. Let $\bm x$ be a stationary point in this neighborhood. If
$x_i>x_i^\star$, stationarity requires $g_i(\bm x)+h'_{i,+}(x_i^\star)=0$
contradicting the first strict inequality. If $x_i<x_i^\star$, stationarity similarly requires $g_i(\bm x)+h'_{i,-}(x_i^\star)=0$
contradicting the second strict inequality. Hence every nearby stationary
point satisfies $\bm x_{Z^\star}
=
\bm x_{Z^\star}^\star.$ If $S^\star=\varnothing$, local uniqueness follows immediately. Otherwise,
shrink the neighborhood further so that
$x_i\in I_i^\star$ for every $i\in S^\star$. Define $\bm{\alpha}_{S^\star}^\star
:=
(h_i'(x_i^\star))_{i\in S^\star}$,
and define the reduced stationarity mapping
\begin{equation}
    \bm R(\bm u)
:=
\bm g_{S^\star}\!\left(
\bm E_{S^\star}\bm u
+
\bm E_{Z^\star}\bm x_{Z^\star}^\star
\right)
+
\bm{\alpha}_{S^\star}^\star.
\end{equation}
Then $\bm R(\bm x_{S^\star}^\star)=\bm 0$ and $\nabla\bm R(\bm x_{S^\star}^\star)
=
\bm H_{S^\star S^\star}(\bm x^\star)$. The reduced Hessian is nonsingular, so the inverse function theorem implies that
$\bm x^\star$ is the unique nearby zero of $\bm R$. Thus
$\bm x^\star$ is isolated in $\mathcal X$.

It remains to prove convergence of the full sequence. The iterates remain
in the compact initial sublevel set. By Theorem~\ref{thm:3}, $\|\bm x_{k+1}-\bm x_k\|\to0$
and every accumulation point belongs to $\mathcal X$. By the standard
connected cluster-set result for bounded sequences with vanishing
successive differences~\cite{ostrowski2016solution}, the cluster set $\Omega$ of
$\{\bm x_k\}$ is nonempty, compact, and connected. Because $\bm x^\star\in\Omega\subseteq\mathcal X$
and $\bm x^\star$ is isolated in $\mathcal X$, the singleton
$\{\bm x^\star\}$ is both relatively open and relatively closed in
$\Omega$. Connectedness therefore implies $\Omega=\{\bm x^\star\}$.
Hence, $\bm x_k\to\bm x^\star.$
\end{proof}
\section{Proof of Theorem~\ref{thm:4}}
\label{appendix-I}
\begin{proof}
Proposition~\ref{prop:4} and Theorems~\ref{the:2} and \ref{thm:3} give $\bm x_k\to\bm x^\star,
~
\bm p_k\to\bm0,
~
\beta_k\le M,
~
\lambda_k\le\bar\lambda.$
Hence $x_{k,i}\in I_i^\star$ for every $i\in S^\star$ and all sufficiently
large $k$. If $Z^\star=\varnothing$, the conclusion follows immediately.
Hence assume $Z^\star\ne\varnothing$. Set $\bm y_k=\bm x_k+\bm p_k.$
Strict complementarity, continuity of $\bm g$, and $\limsup_{k\to\infty}\nu_k<\Delta^\star$
imply that there exist $\varepsilon>0$ and $K_0$ such that, for every
$k\ge K_0$ and $i\in Z^\star$,
\begin{align}
    \min\left\{
g_{k,i}+h'_{i,+}(x_i^\star),
-g_{k,i}-h'_{i,-}(x_i^\star)
\right\}
-\nu_k
&\ge 2\varepsilon, \nonumber
\\
\|\bm e_k\|+M\|\bm p_k\|&<\varepsilon \nonumber.
\end{align}
Increase $K_0$, if necessary, so that neither $x_{k,i}$ nor $y_{k,i}$ can
reach a breakpoint other than $x_i^\star$.

Fix $i\in Z^\star$ and $k\ge K_0$. First suppose $x_{k,i}\ne x_i^\star$
and let $s_{k,i}
:=
\operatorname{sign}(x_{k,i}-x_i^\star).$
Then $i\in S_k$, and the slope of its current interval is $\alpha_{k,i}
=
h'_{i,s_{k,i}}(x_i^\star).$
The preceding margin therefore gives $s_{k,i}(g_{k,i}+\alpha_{k,i})
\ge 2\varepsilon.$ We claim that $y_{k,i}=x_i^\star.$
Otherwise, $y_{k,i}$ lies in the interior of the current affine interval, so
the corresponding normal component is zero. The coordinate residual
equation then gives
\[
\begin{aligned}
2\varepsilon
&\le
s_{k,i}(g_{k,i}+\alpha_{k,i})\le
|e_{k,i}|
+
|(\bm B_k\bm p_k)_i| \\
&\le
\|\bm e_k\|
+
M\|\bm p_k\|
<
\varepsilon,
\end{aligned}
\]
which is impossible. Hence $y_{k,i}=x_i^\star,~
p_{k,i}=x_i^\star-x_{k,i},$
and therefore $i\in D_k$.

Since strict complementarity gives $-g_i(\bm x^\star)
\in
\operatorname{ri}\partial h_i(x_i^\star),$
the convergence above implies, for all sufficiently large $k$,
\[
-g_{k,i}-\beta_kp_{k,i}
\in
\operatorname{ri}\partial h_i(x_i^\star).
\]
Thus the projection in \eqref{eq:vbar} is inactive and $v_{k,i}
=
-g_{k,i}-\beta_kp_{k,i}$.
Because $\operatorname{sign}(p_{k,i})=-s_{k,i}$,
we obtain
\[
\begin{aligned}
\beta_k+\frac{\omega_{k,i}}{|p_{k,i}|}
&=
\frac{
s_{k,i}(g_{k,i}+\alpha_{k,i})
}{
|p_{k,i}|
}\ge
\frac{2\varepsilon}{|p_{k,i}|}
\longrightarrow+\infty.
\end{aligned}
\]
Since $\lambda_k\le\bar\lambda$, 
Theorem~\ref{the:1} eventually gives $x_{k+1,i}
=
y_{k,i}
=
x_i^\star$

Now suppose $x_{k,i}=x_i^\star.$
Both one-sided directional derivatives are larger than $\nu_k$, so $i\notin V_k\cup E_k.$
Since $i$ is a kink coordinate, it follows that $i\notin W_k$ and
$p_{k,i}=0$. \eqref{eq:the:1-1} then gives $x_{k+1,i}=x_{k,i}=x_i^\star.$ Thus each coordinate in $Z^\star$ reaches its target kink after finitely
many iterations and remains there. Since $Z^\star$ is finite, there is a
common $K_{\mathrm{id}}<\infty$ such that $\bm x_{k,Z^\star}
=
\bm x_{Z^\star}^\star$ with $k\ge K_{\mathrm{id}}$.
\end{proof}

\section{Proof of Theorem \ref{thm:5}}
\label{appendix-J}
\begin{proof}
The finite identification property in Theorem \ref{thm:4} implies for sufficiently large $k$, $W_k=S^\star$ and $\bm p_{k,Z^\star}=\bm 0$. Since $\bm{p}_k\to 0$, the model residual becomes
\begin{equation} 
\bm B_k^S\bm p_k^S = -\bm r_k^S+\bm e_k^S. 
\label{eq:local-identified-model-equation} \end{equation}
 Let $\bm s_k := \bm x_{k,S^\star} - \bm x_{S^\star}^\star.$   Since $\bm r_\star^S=\bm 0$, the mean-value formula gives $\bm r_k^S = \bm A_k\bm s_k$, and 
\begin{equation}
    \bm A_k := \int_0^1 \bm H_{S^\star S^\star} \bigl( \bm x^\star + t(\bm x_k-\bm x^\star) \bigr) \,dt.\nonumber
\end{equation}
Because $\bm A_k\to\bm H_\star^S$ and $\bm H_\star^S$ is nonsingular, $\|\bm r_k^S\| = \Theta(\|\bm s_k\|).$ Combining \eqref{eq:local-identified-model-equation}, \eqref{eq:bounded-hessian} and Assumption~\ref{assump:3} gives
\[ (m-o(1))\|\bm p_k^S\| \le \|\bm r_k^S\| \le (M+o(1))\|\bm p_k^S\|. \] 
Therefore, 
\begin{equation} \|\bm p_k\| = \Theta\!\left( \|\bm x_k-\bm x^\star\| \right). \label{eq:local-step-error-equivalence} \end{equation} 
We next show that the full step is eventually accepted. At $\lambda=\beta_k$, $\bm z_k(\beta_k) = \bm x_k+\bm p_k.$ Since $h$ is affine on the identified manifold and $q$ is twice continuously differentiable, 
\[ \begin{aligned} & F(\bm x_k+\bm p_k)-F(\bm x_k) + \frac{\eta\beta_k}{2}\|\bm p_k\|^2 \\ &\quad= -\frac{1-\eta}{2} (\bm p_k^S)^\top \bm B_k^S \bm p_k^S + (\bm e_k^S)^\top\bm p_k^S \\ &\qquad + \frac12 (\bm p_k^S)^\top (\bm H_k^S-\bm B_k^S) \bm p_k^S + o(\|\bm p_k^S\|^2), \end{aligned} \] 
where we used \eqref{eq:local-identified-model-equation} and $\beta_k\|\bm p_k\|^2 = (\bm p_k^S)^\top \bm B_k^S \bm p_k^S$. By Assumption~\ref{assump:3}, the two error terms and the Taylor remainder are $o(\|\bm p_k^S\|^2)$. Hence, \[ \begin{aligned} & F(\bm x_k+\bm p_k)-F(\bm x_k) + \frac{\eta\beta_k}{2}\|\bm p_k\|^2 \\ &\qquad\le - \left( \frac{(1-\eta)m}{2} - o(1) \right) \|\bm p_k^S\|^2 <0 \end{aligned} \] for all sufficiently large $k$. Thus, $\lambda_k = \beta_k$ eventually. For these iterations, $\bm s_{k+1} = \bm s_k+\bm p_k^S.$ Using \eqref{eq:local-identified-model-equation} and $\bm r_k^S=\bm A_k\bm s_k$, we obtain 
\[ \bm H_k^S\bm s_{k+1} = (\bm H_k^S-\bm A_k)\bm s_k + \bm e_k^S + (\bm H_k^S-\bm B_k^S)\bm p_k^S. \] 
Continuity of the Hessian gives $\|\bm H_k^S-\bm A_k\|\to0$. Assumption~\ref{assump:3} and \eqref{eq:local-step-error-equivalence} show that the remaining two terms are also $o(\|\bm s_k\|)$. Since $ \bm H_k^S\to\bm H_\star^S $ and $\bm H_\star^S$ is nonsingular, $(\bm H_k^S)^{-1}$ is uniformly bounded for all sufficiently large $k$. Therefore, $ \|\bm s_{k+1}\| = o(\|\bm s_k\|). $ The coordinates in $Z^\star$ are fixed at their limiting values, so \eqref{eq:thm-5-1} holds.
\end{proof}
\vspace{-0.5em}

\bibliographystyle{IEEEtran}
\bibliography{references}

@article{candes2008enhancing,
  author       = {Cand{\`e}s, Emmanuel J. and Wakin, Michael B. and Boyd, Stephen P.},
  title        = {Enhancing sparsity by reweighted {$\ell_1$} minimization},
  journal      = {J. Fourier Anal. Appl.},
  volume       = {14},
  number       = {5--6},
  pages        = {877--905},
  year         = {2008},
  doi          = {10.1007/s00041-008-9045-x},
  note         = {doi: 10.1007/s00041-008-9045-x},
}

@article{fan2001variable,
  author       = {Fan, Jianqing and Li, Runze},
  title        = {Variable selection via nonconcave penalized likelihood and its oracle properties},
  journal      = {J. Amer. Statist. Assoc.},
  volume       = {96},
  number       = {456},
  pages        = {1348--1360},
  year         = {2001},
  doi          = {10.1198/016214501753382273},
  note         = {doi: 10.1198/016214501753382273},
}

@article{zhang2010nearly,
  author       = {Zhang, Cun-Hui},
  title        = {Nearly unbiased variable selection under minimax concave penalty},
  journal      = {Ann. Statist.},
  volume       = {38},
  number       = {2},
  pages        = {894--942},
  year         = {2010},
  doi          = {10.1214/09-AOS729},
  note         = {doi: 10.1214/09-AOS729},
}

@article{soubies2015continuous,
  author       = {Soubies, Emmanuel and Blanc-F{\'e}raud, Laure and Aubert, Gilles},
  title        = {A continuous exact {$\ell_0$} penalty ({CEL0}) for least squares regularized problem},
  journal      = {SIAM J. Imaging Sci.},
  volume       = {8},
  number       = {3},
  pages        = {1607--1639},
  year         = {2015},
  doi          = {10.1137/151003714},
  note         = {doi: 10.1137/151003714},
}

@article{rybak2025inference,
  author       = {Jakub Rybak and Heather Battey and Wen-Xin Zhou},
  title        = {On inference for the support vector machine},
  journal      = {J. Mach. Learn. Res.},
  volume       = {26},
  number       = {85},
  pages        = {1--54},
  year         = {2025},
}

@inproceedings{wang2024optimal,
  author       = {Caixing Wang and Xingdong Feng},
  title        = {Optimal kernel quantile learning with random features},
  booktitle    = {Proc. 41st Int. Conf. Mach. Learn. ({ICML})},
  series       = {Proc. Mach. Learn. Res.},
  volume       = {235},
  pages        = {50419--50452},
  publisher    = {PMLR},
  year         = {2024},
}

@article{liang2017activity,
  author       = {Liang, Jingwei and Fadili, Jalal and Peyr{\'e}, Gabriel},
  title        = {Activity identification and local linear convergence of forward--backward-type methods},
  journal      = {SIAM J. Optim.},
  volume       = {27},
  number       = {1},
  pages        = {408--437},
  year         = {2017},
  doi          = {10.1137/16M106340X},
  note         = {doi: 10.1137/16M106340X},
}

@article{bareilles2023newton,
  author       = {Bareilles, Gilles and Iutzeler, Franck and Malick, J{\'e}r{\^o}me},
  title        = {{Newton} acceleration on manifolds identified by proximal-gradient methods},
  journal      = {Math. Program.},
  volume       = {200},
  pages        = {37--70},
  year         = {2023},
  doi          = {10.1007/s10107-022-01873-w},
  note         = {doi: 10.1007/s10107-022-01873-w},
}

@article{byrd2016inexact,
  author       = {Byrd, Richard H. and Nocedal, Jorge and Oztoprak, Figen},
  title        = {An inexact successive quadratic approximation method for {$\ell_1$} regularized optimization},
  journal      = {Math. Program.},
  volume       = {157},
  number       = {2},
  pages        = {375--396},
  year         = {2016},
  doi          = {10.1007/s10107-015-0941-y},
  note         = {doi: 10.1007/s10107-015-0941-y},
}

@article{lee2019inexact,
  author       = {Lee, Ching-Pei and Wright, Stephen J.},
  title        = {Inexact successive quadratic approximation for regularized optimization},
  journal      = {Comput. Optim. Appl.},
  volume       = {72},
  number       = {3},
  pages        = {641--674},
  year         = {2019},
  doi          = {10.1007/s10589-019-00059-z},
  note         = {doi: 10.1007/s10589-019-00059-z},
}

@article{kanzow2021globalized,
  author       = {Kanzow, Christian and Lechner, Theresa},
  title        = {Globalized inexact proximal {Newton}-type methods for nonconvex composite functions},
  journal      = {Comput. Optim. Appl.},
  volume       = {78},
  number       = {2},
  pages        = {377--410},
  year         = {2021},
  doi          = {10.1007/s10589-020-00243-6},
  note         = {doi: 10.1007/s10589-020-00243-6},
}

@book{ostrowski2016solution,
  author       = {Ostrowski, Alexander M.},
  title        = {Solution of Equations and Systems of Equations},
  series       = {Pure and Applied Mathematics: A Series of Monographs and Textbooks},
  volume       = {9},
  publisher    = {Elsevier},
  year         = {2016},
}

@article{bonettini2016variable,
  author       = {Bonettini, Silvia and Loris, Ignace and Porta, Federica and Prato, Marco},
  title        = {Variable metric inexact line-search-based methods for nonsmooth optimization},
  journal      = {SIAM J. Optim.},
  volume       = {26},
  number       = {2},
  pages        = {891--921},
  year         = {2016},
  doi          = {10.1137/15M1019325},
  note         = {doi: 10.1137/15M1019325},
}

@article{ochs2019adaptive,
  author       = {Ochs, Peter and Pock, Thomas},
  title        = {Adaptive {FISTA} for nonconvex optimization},
  journal      = {SIAM J. Optim.},
  volume       = {29},
  number       = {4},
  pages        = {2482--2503},
  year         = {2019},
  doi          = {10.1137/17M1156678},
  note         = {doi: 10.1137/17M1156678},
}

@article{liang2023average,
  author       = {Liang, Jiaming and Monteiro, Renato D. C.},
  title        = {Average curvature {FISTA} for nonconvex smooth composite optimization problems},
  journal      = {Comput. Optim. Appl.},
  volume       = {86},
  pages        = {275--302},
  year         = {2023},
  doi          = {10.1007/s10589-023-00490-3},
  note         = {doi: 10.1007/s10589-023-00490-3},
}

@article{bonettini2024heavyball,
  author       = {Bonettini, Silvia and Prato, Marco and Rebegoldi, Simone},
  title        = {A new proximal heavy ball inexact line-search algorithm},
  journal      = {Comput. Optim. Appl.},
  volume       = {88},
  number       = {2},
  pages        = {525--565},
  year         = {2024},
  doi          = {10.1007/s10589-024-00565-9},
  note         = {doi: 10.1007/s10589-024-00565-9},
}

@article{hare2004identifying,
  author       = {Hare, Warren L. and Lewis, Adrian S.},
  title        = {Identifying active constraints via partial smoothness and prox-regularity},
  journal      = {J. Convex Anal.},
  volume       = {11},
  number       = {2},
  pages        = {251--266},
  year         = {2004},
}

@article{vomdahl2024inexact,
  author       = {vom Dahl, Simeon and Kanzow, Christian},
  title        = {An inexact regularized proximal {Newton} method without line search},
  journal      = {Comput. Optim. Appl.},
  volume       = {89},
  number       = {3},
  pages        = {585--624},
  year         = {2024},
  doi          = {10.1007/s10589-024-00600-9},
  note         = {doi: 10.1007/s10589-024-00600-9},
}

@article{yao2018efficient,
  author       = {Yao, Quanming and Kwok, James T.},
  title        = {Efficient learning with a family of nonconvex regularizers by redistributing nonconvexity},
  journal      = {J. Mach. Learn. Res.},
  volume       = {18},
  number       = {179},
  pages        = {1--52},
  year         = {2018},
}

@electronic{libsvmdata,
  author       = {Chih-Chung Chang and Chih-Jen Lin},
  title        = {{LIBSVM} data: Classification, regression, and multi-label},
  note         = {Accessed: Aug. 28, 2026},
  url          = {https://www.csie.ntu.edu.tw/~cjlin/libsvmtools/datasets/},
}

@article{cheng2021active,
  author       = {Cheng, Wanyou and Dai, Yu-Hong},
  title        = {An active set {Newton-CG} method for {$\ell_1$} optimization},
  journal      = {Appl. Comput. Harmon. Anal.},
  volume       = {50},
  pages        = {303--325},
  year         = {2021},
  doi          = {10.1016/j.acha.2019.08.005},
  note         = {doi: 10.1016/j.acha.2019.08.005},
}

@article{aravkin2022trustregion,
  author       = {Aravkin, Aleksandr Y. and Baraldi, Robert and Orban, Dominique},
  title        = {A proximal quasi-{Newton} trust-region method for nonsmooth regularized optimization},
  journal      = {SIAM J. Optim.},
  volume       = {32},
  number       = {2},
  pages        = {900--929},
  year         = {2022},
  doi          = {10.1137/21M1409536},
  note         = {doi: 10.1137/21M1409536},
}

@article{baraldi2023trustregion,
  author       = {Baraldi, Robert J. and Kouri, Drew P.},
  title        = {A proximal trust-region method for nonsmooth optimization with inexact function and gradient evaluations},
  journal      = {Math. Program.},
  volume       = {201},
  number       = {1--2},
  pages        = {559--598},
  year         = {2023},
  doi          = {10.1007/s10107-022-01915-3},
  note         = {doi: 10.1007/s10107-022-01915-3},
}

@article{lee2023accelerating,
  author       = {Lee, Ching-Pei},
  title        = {Accelerating inexact successive quadratic approximation for regularized optimization through manifold identification},
  journal      = {Math. Program.},
  volume       = {201},
  number       = {1--2},
  pages        = {599--633},
  year         = {2023},
  doi          = {10.1007/s10107-022-01916-2},
  note         = {doi: 10.1007/s10107-022-01916-2},
}

@inproceedings{andrew2007scalable,
  author       = {Andrew, Galen and Gao, Jianfeng},
  title        = {Scalable training of {$\ell_1$}-regularized log-linear models},
  booktitle    = {Proc. 24th Int. Conf. Mach. Learn. ({ICML})},
  pages        = {33--40},
  publisher    = {ACM},
  year         = {2007},
  doi          = {10.1145/1273496.1273501},
  note         = {doi: 10.1145/1273496.1273501},
}

@article{byrd2016family,
  author       = {Byrd, Richard H. and Chin, Gillian M. and Nocedal, Jorge and Oztoprak, Figen},
  title        = {A family of second-order methods for convex {$\ell_1$}-regularized optimization},
  journal      = {Math. Program.},
  volume       = {159},
  number       = {1--2},
  pages        = {435--467},
  year         = {2016},
  doi          = {10.1007/s10107-015-0965-3},
  note         = {doi: 10.1007/s10107-015-0965-3},
}

@article{delosreyes2017secondorder,
  author       = {{De Los Reyes}, J. C. and Loayza, E. and Merino, P.},
  title        = {Second-order orthant-based methods with enriched {Hessian} information for sparse {$\ell_1$}-optimization},
  journal      = {Comput. Optim. Appl.},
  volume       = {67},
  number       = {2},
  pages        = {225--258},
  year         = {2017},
  doi          = {10.1007/s10589-017-9891-z},
  note         = {doi: 10.1007/s10589-017-9891-z},
}

@article{lewis2021active,
  author       = {Lewis, Adrian S. and Wylie, Calvin},
  title        = {Active-set {Newton} methods and partial smoothness},
  journal      = {Math. Oper. Res.},
  volume       = {46},
  number       = {2},
  pages        = {712--725},
  year         = {2021},
  doi          = {10.1287/moor.2020.1075},
  note         = {doi: 10.1287/moor.2020.1075},
}

@article{keskar2016secondorder,
  author       = {Keskar, Nitish Shirish and Nocedal, Jorge and {\"O}ztoprak, Figen and W{\"a}chter, Andreas},
  title        = {A second-order method for convex {$\ell_1$}-regularized optimization with active-set prediction},
  journal      = {Optim. Methods Softw.},
  volume       = {31},
  number       = {3},
  pages        = {605--621},
  year         = {2016},
  doi          = {10.1080/10556788.2016.1138222},
  note         = {doi: 10.1080/10556788.2016.1138222},
}

@article{chen2017reduced,
  author       = {Chen, Tianyi and Curtis, Frank E. and Robinson, Daniel P.},
  title        = {A reduced-space algorithm for minimizing {$\ell_1$}-regularized convex functions},
  journal      = {SIAM J. Optim.},
  volume       = {27},
  number       = {3},
  pages        = {1583--1610},
  year         = {2017},
  doi          = {10.1137/16M1062259},
  note         = {doi: 10.1137/16M1062259},
}
\end{document}